\documentclass[12pt]{article}

\usepackage[margin=1in]{geometry}
\usepackage[T1]{fontenc}
\usepackage[utf8]{inputenc}
\usepackage{amsmath,amsthm,amssymb}
\usepackage{mathptmx}
\usepackage{bm}
\usepackage{mathdots}
\usepackage{graphicx}
\usepackage{placeins}
\usepackage[authoryear,round]{natbib}
\usepackage{microtype}
\usepackage{setspace}
\usepackage[hidelinks]{hyperref}

\newtheorem{conjecture}{Conjecture}
\newtheorem{proposition}{Proposition}
\newtheorem{lemma}{Lemma}

\newcommand{\OMIT}[1]{}

\title{Bifurcation mechanism at a sustain point\\of a long narrow economy}

\author{%
K. Ikeda\thanks{Department of Civil and Environmental Engineering, Tohoku University, Aoba, Sendai 980-8579, Japan. Email: kiyohiro.ikeda.b4@tohoku.ac.jp.}\\
H. Aizawa\thanks{Department of Civil and Environmental Engineering, Chuo University, Kasuga, Bunkyo-ku, Tokyo 112-8551, Japan. Email: haizawa018@g.chuo-u.ac.jp.}\\
J. M. Gaspar\thanks{School of Economics and Management and CEF.UP, University of Porto, Porto, Portugal. Email: jgaspar@fep.up.pt.}%
}
\date{}

\begin{document}
\maketitle

\begin{abstract}
We investigate population agglomeration in a long narrow economy, in which
an odd number of places are evenly distributed 
over a line segment.
The bifurcation analysis of this economy elucidates
the mechanism of the emergence of twin cities around the central city.
The validity and usefulness of this analysis are confirmed using several
well-known
economic geography models that display
various kinds of bifurcation behaviors.
By this analysis, we investigate the historical change in the
population distribution
in a chain of cities on Japan's Main Island.
\end{abstract}


\section{Introduction}\label{SecIntro}

Chains of cities prosper worldwide, e.g.,
on Japan's Main Island
and in a closed, narrow corridor between the Atlantic Ocean 
and the Appalachian Mountains (see Fig.~\ref{USA=JAPAN}).
In turn, the mechanism of the growth or decay of a large central city, such as 
Tokyo and New York City,
among a chain of cities, is of special importance.    
Figure \ref{USA=JAPAN}(a) shows the populations of five large cities
on Japan's Main Island in the 1950s (shown by blue arcs).
The population in 2020 (shown by orange arcs)
displays significant growth in Nagoya
located in the middle of 
two large cities (Tokyo and Osaka).
This example demonstrates the importance of the geographical advantage of the central city for the agglomeration of
population  in the real world.

\begin{figure}
	\centering
	\small
			\begin{tabular}{c@{\hspace{15mm}}c}
		\includegraphics[scale=0.26]{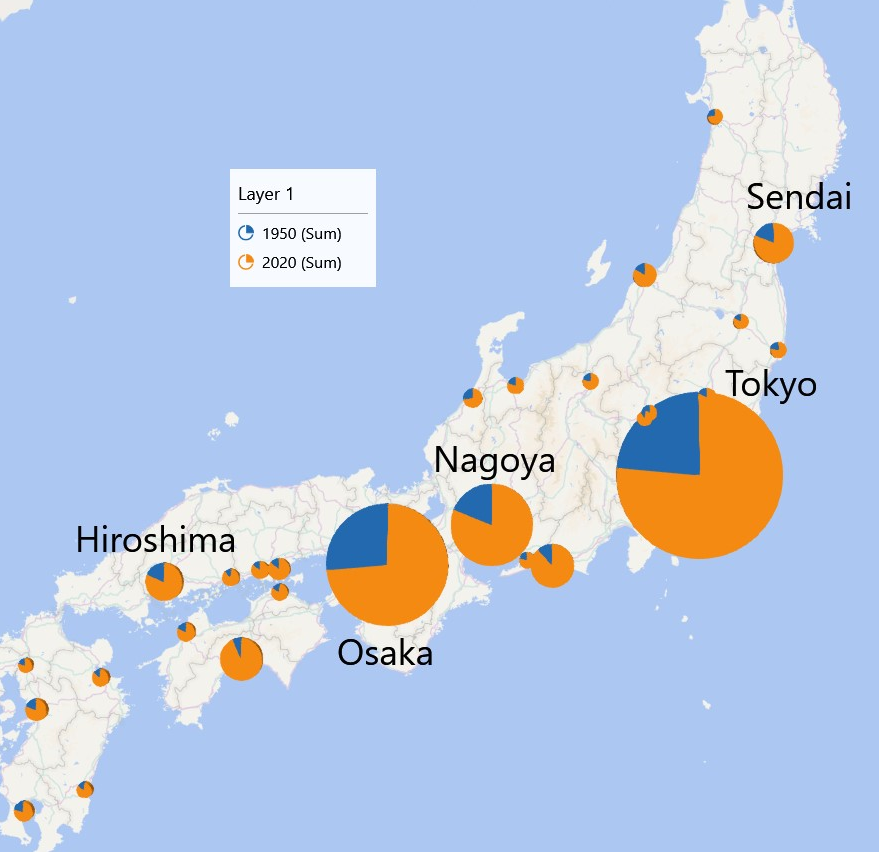} &
			\includegraphics[scale=0.24]{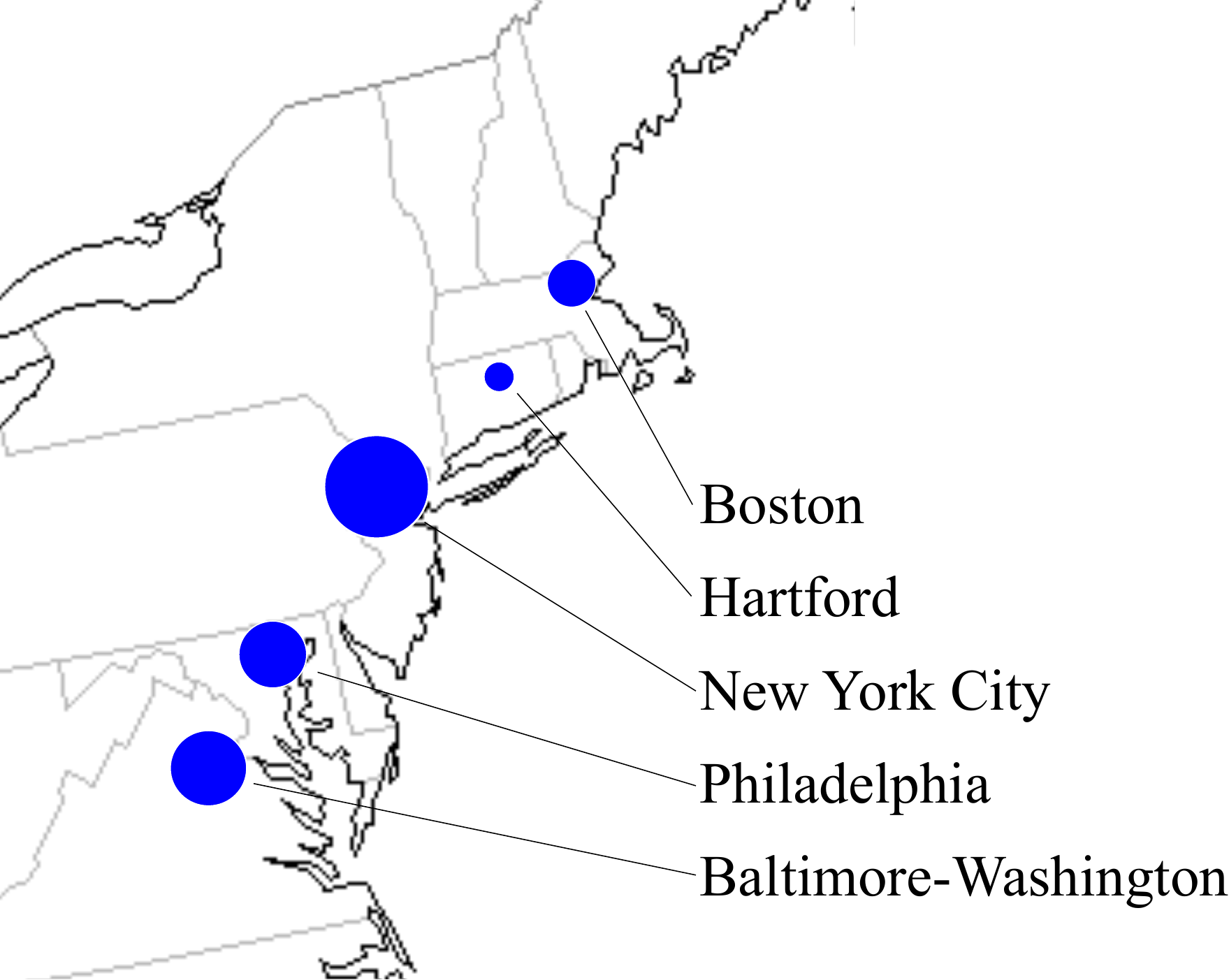} \\
			(a) Japan's Main Island & (b) The East Atlantic \\
		\end{tabular}
	\caption{A chain of cities in the world
		(in (a): the blue arc is the population in 1950 
            and the orange arc 
		is the population in 2020 \cite{united.nations.2022}}
	\label{USA=JAPAN}
\end{figure}

That said, this paper aims to elucidate the 
agglomeration mechanism of the population in 
multiple cities along a narrow corridor.
As a model of such cities, we employ a long narrow economy, 
i.e., an odd number of 
places evenly distributed over a line segment.
Although the real world is two-dimensional, several chains of cities are characterized by a linear arrangement of urban areas. They can be found at transnational scales (e.g., the STRING and Golden, Blue, and Green Bananas in Europe).
Such chains of cities seem to be characterized by discrete agglomerations
because not all areas are suitable for residential locations.
For instance, over 70\% of
Japanese territory consists of largely uninhabited mountains
and rugged uplands \cite{Sargent.1980}.
There, accordingly, are spatial discontinuities between cities with significant population densities.
 Moreover, as argued by  \citet{Akamatsu.et.al.2017}, any model subject to empirical validation must be discrete, because all collected data is aggregated over some geographical areas. 

We answer the question ^^ ^^ How
and where do peripheral cities grow or decay around a central city?" 
For peripheral cities, 
we refer to a pair of cities placed symmetrically around the central city.
These cities form a spatial pattern that resembles a megalopolis.
To tackle this mission,
we elucidate the bifurcation/agglomeration mechanism
of the long narrow economy 
in the following two steps:
\begin{enumerate}
\item
The bifurcation mechanism for general 
economic geography models.
\item
The agglomeration mechanism for specific,
well-known economic geography models.
\end{enumerate}

In the first step,
we study the bifurcation of the full agglomeration,
for which a central city gathers all the population.
This bifurcation occurs
when transport costs (freeness of trade) reach a critical level,
called \textit{sustain point} in economic geography \cite{Fujita.etal.1999}.
 Above (below) this level, this state 
 becomes economically unsustainable, and new cities emerge
 around the large central city. 
While most previous studies investigate the bifurcation of
a uniform distribution, we do the opposite, focusing on what happens once the full agglomeration becomes unsustainable.

The applicability and usefulness of this theoretical study 
are demonstrated by referring to 
well-known economic geography models \cite{Forslid.Ottaviano.2003,Murata.Thisse.2005,Pfluger.Sudekum.2008} (called 
FO, MT, and PFSU models, respectively).
We choose these three models based on the classification regarding the spatial scale of dispersion forces presented by 
\citet{Akamatsu.et.al.2023}; local dispersion forces occur within a region, whereas global dispersion forces depend on the proximity structure between regions. The FO model exhibits only a global dispersion force, whereas the MT model contains only a local one. The PFSU model includes both.
Models with only local dispersion forces cannot explain the existence of multi-peaked agglomerations \cite{Akamatsu.et.al.2023}. 
Hence, the MT model displays a relatively undiversified bifurcation behavior compared to the other two models.

In the second step, we 
study the location of peripheral cities,
based on the analysis of the FO and the PFSU models.	
The MT model is not suited for this study
 since peripheral cities are always located next to 
the central city.
The long narrow economy with five cities 
is employed to investigate the historical change of the 
population distribution
in a chain of cities on Japan's Main Island
(cf., Fig.~\ref{USA=JAPAN}(a)).
This analysis demonstrates the usefulness 
of the study of the long narrow economy.

This paper is organized as follows.
Section~\ref{Related Section} presents related studies.
The bifurcation from the full agglomeration 
to the central city
 is advanced in Section~\ref{Bifurcation mechanism}.
Section~\ref{Economic geography models} 
introduces economic geography models.
 The bifurcation behavior of the three economic geography models
 is studied in Section~\ref{Stability of full agglomeration}.
Bifurcation analysis for the FO model
is advanced in Section~\ref{AnalysisGeography}.
Section~\ref{Use of EG analysis} 
investigates the historical change in the
population distribution
in a chain of cities on Japan's Main Island.
Section~\ref{section conclusion} concludes.


\section{Related studies}\label{Related Section}

\citet{Fujita.Krugman.1995}
investigated the mechanism of the
formation of new cities for a full agglomeration
 in continuous space.
\citet{Allen.Arkolakis.2014} presented 
examples of spatial patterns and their stability by changing trade costs and other parameters.
 However, these works are silent on the exact location of such places around the central city.

A racetrack economy (equally spread places around a circle) is employed in several studies to investigate spatial population patterns theoretically.
Agglomeration processes
in this economy with a core--periphery model 
was studied \cite{Akamatsu.et.al.2012,Ikeda.et.al.JEDC.2012}
and this study was later extended to various kinds of models
\cite{Akamatsu.et.al.2023}.
However, this economy is an idealized spatial platform
that lacks border effects and the locational advantage of a central city.

Some researches report
several agglomeration patterns of the long-narrow economy:
the simplest core--periphery pattern for three places
\cite{Ago.Isono.Tabuchi.2006},
a chain of spatially repeated core--periphery 
patterns \textit{a la} Christaller and L\"{o}sch
\cite{Fujita.Mori.1997},
and a megalopolis which consists of large core cities
that are connected by \textit{an industrial belt}
\cite{Mori.1997}.
These patterns were numerically observed, 
starting from the uniform distribution 
and changing agglomeration forces and transport costs
\cite{IJET.2017}.
However, the agglomeration of population has already emerged in the real world. Therefore, it seems more pertinent 
to start from the full agglomeration to the central city
in the investigation of the rise or fall of peripheral cities. 

Dynamic analyses of 
economic geography models 
displayed chaos and various kinds of bifurcations
\cite{Currie.Kubin.2006,Agliari.2014,Commendatore.2017}.
The bifurcation theory was applied to
dynamic problems in economics
\cite{Dercole.2008,Dercole.2020}.
 There are several books on chaos in economics
 \cite{Janeway.1989,Dechert.1996,Rosser Jr.2000}.

This paper focuses on 
a static analysis as is conventional in economic geography.
Moreover, 
as a novel contribution to economic geography,
the paper elucidates the mechanism of
the formation of satellite cities in a chain of cities
for several well-known models  
and offers an application for real data,
while a recent study of 
such formation on a hexagonal lattice
 focuses only on a theoretical aspect
 \cite{Ikeda.NETS.2023}.


\section{General bifurcation mechanism of a long narrow economy}%
\label{Bifurcation mechanism}

We employ a long narrow economy
composed of an odd number of places as a spatial platform for economic geography models.
We study the bifurcation mechanism 
of the emergence of peripheral cities 
around a large central city.
This study is applicable to  
general economic geography models.

\subsection{Modeling of the spatial economy}\label{Spatial economy model}

The long narrow economy has $K=2k+1$ $(k=1,2,\ldots)$ places 
labeled $i \in N=\left\{-k,...,0,...,k \right\}$,
and are evenly distributed over a line segment
(Fig.~\ref{Linear_space_2}).
The $0$th place is located at the center,
and a place $i \neq 0$ is  
$| i |$ \textit{steps} away from the center. 

\begin{figure}[b]
	\begin{center}
		\includegraphics[scale=0.45]{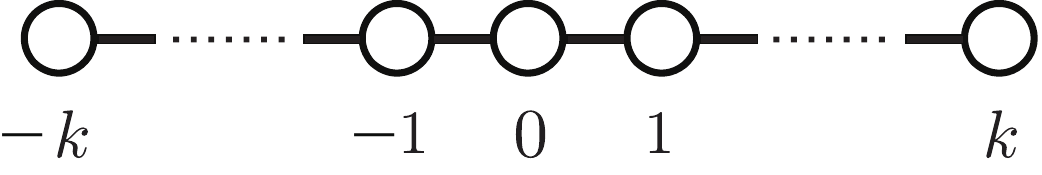}
	\end{center}
	\caption{A long narrow economy}
	\label{Linear_space_2}
\end{figure}

There are mobile agents
that migrate among places. 
The number of these agents in place $i$ is denoted by $\lambda_i$
and is subject to the conservation law: 
\begin{align}\label{ConservationLaw} 
 & \sum_{i\in N} \lambda_i =1 .
 \end{align}
Agents choose to live in the place with the highest
indirect utility $v_i$.
We assume that
the indirect utility $v_i$
is a sufficiently smooth nonlinear function of the population distribution 
$\bm{\lambda}=(\lambda_i \mid i\in N)$ and model parameters.
Furthermore, we assume that 
there is no initial heterogeneity in 
each pair of 
places $i=\pm j$  $(\neq 0)$
regarding the parameters of $v_i$.

We consider 
the replicator dynamics 
\cite{TaylorJonker1978}:
$\frac{d\bm{\lambda}}{dt} = 
{\bm F}(\bm{\lambda},\phi)$,
where 
${\bm F}(\bm{\lambda},\phi)=(F_i(\bm{\lambda},\phi)\mid i\in N)$
and
\begin{equation}\label{F_expression}
F_i(\bm{\lambda},\phi)=
(v_i(\bm{\lambda},\phi) - \bar{v}(\bm{\lambda},\phi)) \lambda_i,
\quad i\in N.
\end{equation}

\noindent
Here,
$\bar{v}=\sum_{i\in N} \lambda_i v_i$
expresses the weighted average utility and $\phi \in (0,1) $ is 
\textit{trade freeness}, which is an inverse measure of transportation cost.
 We choose the trade freeness as the bifurcation parameter 
 that expresses the historical tendency of decreasing/increasing transport costs, as is customary in economic geography.
Stationary point 
$({\bm \lambda},\phi)$
of the dynamics is defined as 
a solution of the governing equation:
\begin{equation}\label{gov.eq.general.F}
 \bm{F}({\bm \lambda},\phi)={\bf 0}.
\end{equation}

Figure~\ref{Examples_FA_Twin}
depicts important agglomeration patterns
in economic geography: 
full agglomeration,
twin cities, uniform, core--periphery, and diffused patterns
(cf., Section~\ref{AgglomerationBehaviorsMore}).

\begin{figure}
\begin{small}
	\begin{center}
	\begin{tabular}{c@{\hspace{10mm}}c@{\hspace{5mm}}c}
		\includegraphics[scale=0.60]{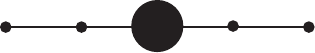} &
		\includegraphics[scale=0.60]{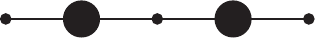} &
		\includegraphics[scale=0.60]{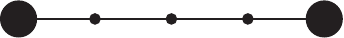} \\
		\noalign{\vskip 1ex}
 (a) Full agglomeration $\bm{\lambda}^{\rm FA}_0$ &
     \multicolumn{2}{c}{(b) Twin cities}
		\end{tabular}

\vspace{4mm}
	\begin{tabular}{c@{\hspace{10mm}}c@{\hspace{10mm}}c}
		\includegraphics[scale=0.60]{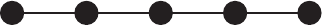} &
		\includegraphics[scale=0.60]{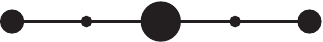} &
		\includegraphics[scale=0.60]{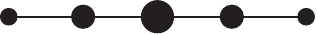} \\
		\noalign{\vskip 1ex}
		(c) Uniform &
		(d) Core--periphery &
		(e) Diffused 
		\end{tabular}
     \end{center}
	\end{small}
	\caption{Agglomeration patterns in a long narrow economy
	for $K=5$ places}
	\label{Examples_FA_Twin}
\vspace{0mm}
\end{figure}

\subsection{Invariant patterns}\label{Invariant long narrow}

There is a unique agglomeration pattern
$\bar{\bm{\lambda}}$
which is a solution of the governing equation \eqref{gov.eq.general.F}
for any values of the trade freeness $\phi$
and economic parameters of the models.
This pattern, called an \textit{invariant pattern} 
\cite{Ikeda.et.al.JEDC.2012,JEDC.2018,Ikeda.Chaos.2019},
plays an important role in the study of agglomeration patterns.

Full agglomeration and twin cities defined below are invariant patterns
(Proposition~\ref{Invariant proposition}).
\begin{itemize}
\item
The full agglomeration 
$\bm{\lambda}^{\rm FA}_j$
to the $j$th place,
i.e.,
$\lambda_{j}=1$ for some $j \in N$
and no population elsewhere.
\item
The twin cities
$\bm{\lambda}^{\rm Twin}_j$ 
in places $i=\pm j$ that gather all population,
i.e., 
$\lambda_{\pm j}=1/2$ for some $j \in \{1,\ldots,k\}$
and no population elsewhere. 
\end{itemize}
Then, we have the following proposition.

\begin{proposition}\label{Invariant proposition}
Full agglomeration 
$\bm{\lambda}^{\rm FA}_{j}$
and twin cities
$\bm{\lambda}^{\rm Twin}_j$ 
are invariant patterns.
\end{proposition}
\begin{proof}
See Appendix~\ref{Proof FA twin} for the proof.
\end{proof}

\subsection{Bifurcation of the full agglomeration at the center}%
\label{Bifurcation theory full}

We elucidate the bifurcation mechanism
of the full agglomeration $\bm{\lambda}^{\rm FA}_{0}$  at the center
to search for the locations of emerging peripheral
cities around the large central city. 
We chose this full agglomeration that has
superior stability to other full agglomerations
(Section~\ref{Other than center}).

We search for stable solutions 
with positive populations in the central and peripheral cities
by the following two steps:
(1) bifurcation analysis to find bifurcating paths, and
(2) stability analysis to identify stable paths.

A solution $(\bm{\lambda},\phi)$ with no population 
($\lambda_i=0$) in one or more places is called a corner solution.
This solution
is called sustainable
if $v_i - \bar{v} \le 0$ for all $i\in N$
and is unsustainable 
if $v_i - \bar{v} > 0$ for some $i\in N$. 
The full agglomeration $(\bm{\lambda}^{\rm FA}_{0},\phi)$ 
is a corner solution,
and its sustainability is described below.

\begin{itemize}
\item
A \textit{local sustain point} 
for the $i$th place is defined by
$\phi=\phi_i^{\rm s}$
where $\lambda_i=0$
and $v_i -\bar{v}= 0$ hold.
\item
A \textit{sustain point } $\phi=\phi^{\rm s}$
for the whole system is defined as 
a particular local sustain point for some $i=\pm \delta_{\rm s}$
satisfying $v_{\delta_{\rm s}} -\bar{v}= 0$ 
and $v_{i} -\bar{v}< 0$ 
for all $i\neq \pm \delta_{\rm s}$.
\end{itemize}
The places $i = \pm \delta_{\rm s}$
for the sustain point 
are highlighted as the locations of emerging peripheral cities.

\subsubsection{Bifurcation analysis}

We consider a local sustain point for places $i=\pm j$ $(\neq 0)$,
satisfying $v_{- j}-v_0=v_{j}-v_0=0$. 
This sustain point is a bifurcation point,
from which one or two peripheral
cities emerge, $j$ steps away from the central place,
as explained below.
In the neighborhood of this point,
at most three components  
$\lambda_{-j},~ \lambda_0, ~\lambda_{j}$ are nonzero.
Then $\bm{F}({\bm \lambda},\phi)={\bf 0}$ in \eqref{gov.eq.general.F}
with the replicator form \eqref{F_expression}
reduces to the so-called bifurcation equation 
\begin{align}\label{Eqs pm j}
\begin{cases}
  [v_{-j}(\lambda_{-j},\lambda_{j},\phi)
   - \bar{v}(\lambda_{-j},\lambda_{j},\phi)] \lambda_{-j}=0,
  \\
 [v_{j}(\lambda_{-j},\lambda_{j},\phi)
   - \bar{v}(\lambda_{-j},\lambda_{j},\phi)] \lambda_{j}=0
\end{cases}
\end{align}
for two nonzero components  
$\lambda_{-j}$ and $\lambda_{j}$.
By \eqref{ConservationLaw},
the other nonzero component is given as
 $\lambda_0=1-\lambda_{-j}-\lambda_{j}$ and 
$[v_{0}(\bm{\lambda},\phi) - \bar{v}(\bm{\lambda},\phi)] \lambda_{0}=0$
is satisfied.
By the bilateral symmetry of the places $-j$ and $j$,
we have the symmetry conditions:
\begin{align}\label{symmetry conditions App}
\begin{cases}
 \bar{v}(\lambda_{j},\lambda_{-j},\phi)=\bar{v}(\lambda_{-j},\lambda_j,\phi), \\
 v_{-j}(\lambda_{j},\lambda_{-j},\phi)=v_j (\lambda_{-j},\lambda_j,\phi).
\end{cases}
\end{align}

We have the following proposition for possible solutions to \eqref{Eqs pm j}.
For $K=5$ places, these solutions are depicted in Fig.~\ref{Blanches}. 

\begin{figure}[b]
	\begin{center}
		\includegraphics[scale=0.45]{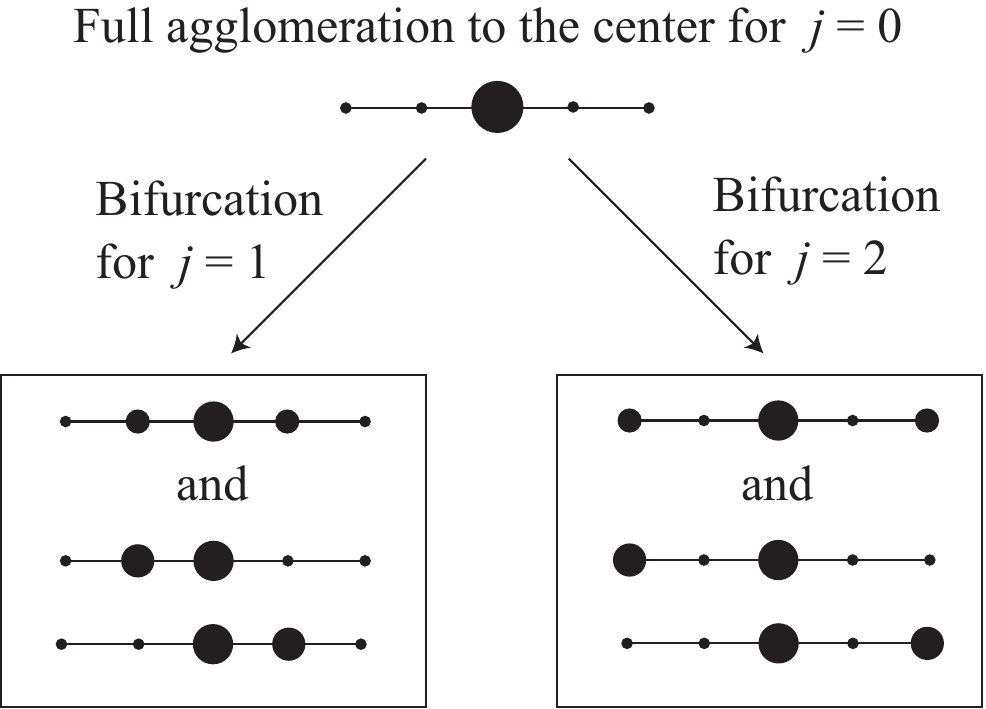} 
	\end{center}
	\caption{Possible bifurcations for $K=5$ places}
	\label{Blanches}
\end{figure}

\begin{proposition}\label{BifurCurveV_2}
At a local sustain point of $\bm{\lambda}^{\rm FA}_{0}$
with $v_j-\bar{v}=0$,
we have three kinds of solution curves: 
\begin{align*}
\begin{cases}
\mbox{Full agglomeration $(\lambda_{-j}=\lambda_{j}=0,~\lambda_0=1)$.} \cr
\mbox{Bifurcating path with two identical peripheral cities
with $\lambda_{j}=\lambda_{-j}>0$}. 
 \cr
\mbox{Bifurcating path with one peripheral city
with $\lambda_{j}>\lambda_{-j}=0$ 
or $\lambda_{-j}>\lambda_{j}=0$.} \cr
\end{cases}
\end{align*}
\end{proposition}
\begin{proof}
See Appendix~\ref{Proof Solutions}
for the proof.
\end{proof}

\subsubsection{Stability analysis}

At a sustain point, a stable bifurcating path can branch
 as described by the following proposition.

\begin{proposition}\label{Stability_Bifcurve}
At a sustain point of $\bm{\lambda}^{\rm FA}_{0}$,  
there is at most one bifurcating path
that is asymptotically stable.
\end{proposition}
\begin{proof}
See Lemma~\ref{Stability_Bifcurve Lemma}
in Appendix~\ref{Stability Local Sustain FA}
for the proof.
\end{proof}

We explain how a stable bifurcating path branches
at a sustain point referring to 
an image of this bifurcation in 
Fig. \ref{Branches}.
The horizontal line SS$^\prime$ denotes
the state of stable (sustainable)
full agglomeration $(\phi>\phi^{\rm s})$.
Two bifurcating paths, SA and SB, branch at the sustain point S denoted by $(\circ)$.
The path SA 
residing on the opposite side of
SS$^\prime$ ($\phi<\phi^{\rm s}$) is
called ^^ ^^ branching forward" and the path SB 
residing on the same side of
SS$^\prime$ $(\phi>\phi^{\rm s})$ is
called ^^ ^^ branching backward".
Then,
we have the following proposition.

\begin{proposition}\label{Stability_of_Bifcurve}
At a sustain point of $\bm{\lambda}^{\rm FA}_{0}$,
if a stable bifurcating path exists, it branches forward.
\end{proposition}
\begin{proof}
See
Lemma~\ref{Stability_Bifcurve Lemma}(iv)
in Appendix~\ref{Stability Local Sustain FA}
for the proof.
\end{proof}

\noindent
Just after bifurcation at point S 
in Fig.~\ref{Branches},
the bifurcating path SB is always unstable
while the stability of the path SA depends on the model and its economic parameters.
Thus, a stable bifurcating path, when it exists, 
realizes a continuation of stable paths
as $\phi$ crosses the sustain point $\phi=\phi^{\rm s}$,
connecting the stable core--periphery pattern AS and
the stable full agglomeration SS$^\prime$.
We show examples of the stable bifurcating paths for the 
FO and PFSU models in Section~\ref{Bifurcation from a full}.
In the literature, 
such continuation is 
observed for the PF model 
 with the two places economy \cite{Pfluger.2004}  and with the multi-region equidistant economy \cite{gaspar.2018,gaspar.2021}.

\begin{figure}[htbp]
	\begin{center}
		\includegraphics[scale=0.45]{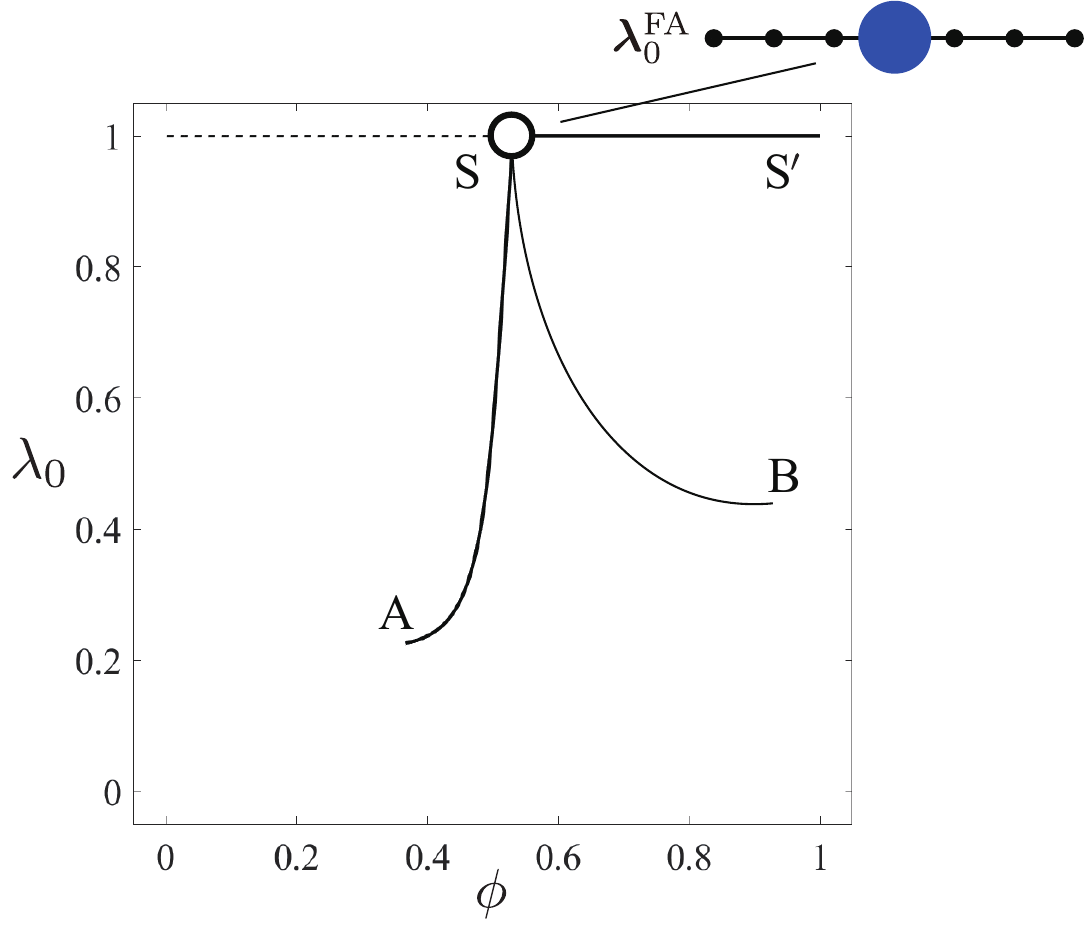}
	\end{center}
	\caption{Image of Proposition \ref{Stability_of_Bifcurve} for the case that 
path SS$^\prime$ for the full agglomeration is stable. } 
	\label{Branches}
\end{figure}

\subsection{Bifurcation 
of the twin cities}\label{Twin-Bifurcation}

For the twin cities $\bm{\lambda}^{\rm Twin}_j$, 
we investigate the bifurcation of the local sustain point.
Eigenvalues of the Jacobian matrix of the twin cities
are $v_0-\bar{v}$ (repeated once) and
 $v_{-\delta}-\bar{v}=v_{\delta}-\bar{v}$ (repeated twice 
 with $v_{-\delta}=v_{\delta}$ and
 $\bar{v}=v_{-j}=v_{j}$;
 $\delta \in \{1,\ldots,k\},\delta\neq j)$.
Accordingly, there are two kinds of sustain points
associated with either (1) $v_0-\bar{v}=0$
or (2) $v_{\pm \delta}-\bar{v}=0$. 

The bifurcation at a sustain point with $v_0-\bar{v}=0$
produces a central city.
The growth of the central city produces a core--satellite pattern
(Fig.~\ref{Examples_FA_Twin}(d)), 
as observed in the 
FO model (Section~\ref{AgglomerationBehaviorsMore}). 

\begin{proposition}\label{twin sustain Bifurcation}
At a sustain point of
$\bm{\lambda}^{\rm Twin}_{j}$ 
with $v_0-\bar{v}=0$,
we have:

\noindent
(i) There is only one bifurcating path,
which has positive populations 
in three places $i=0$ and $\pm j$
($\lambda_{-j}=\lambda_{j}$).

\noindent
(ii) If a stable bifurcating path exists, it branches forward.
\end{proposition}
\begin{proof}
This bifurcation point is a transcritical bifurcation point
and its bifurcation mechanism and stability are studied in the literature
\cite[Section 2.5.2]{Ikeda.Murota.2019}.
\end{proof}

For the other sustain point 
associated with $v_{\pm \delta}-\bar{v}=0$ $(\delta\in \{1,\ldots,k\})$,
we have the following proposition,
similarly to Propositions~\ref{BifurCurveV_2}, \ref{Stability_Bifcurve},
and \ref{Stability_of_Bifcurve}
for the full agglomeration $\bm{\lambda}^{\rm FA}_{0}$. 

\begin{proposition}\label{twin sustain double}
At a sustain point of 
$\bm{\lambda}^{\rm Twin}_{j}$ 
		with $v_{\pm \delta}-\bar{v}=0$
				$(\delta \in \{1,\ldots,k\},~\delta\neq j)$,
		we have: 

\noindent
(i) A bifurcating solution 
	has positive populations only in
		a pair of twin places $i=\pm j$ and $\pm \delta$.
		
		\noindent
(ii) The other solution has positive populations in three places  
		$i=\pm j$ and $i=-\delta$ or $\delta$.
		
		\noindent
(iii) If a stable bifurcating path exists, it branches forward.
\end{proposition}


\section{Economic geography models}\label{Economic geography models}

We employ three well-known economic geography models, namely the 
FO model
\cite{Forslid.Ottaviano.2003},
the PFSU model 
\cite{Pfluger.Sudekum.2008},
and the MT model \cite{Murata.Thisse.2005}.
We carefully choose these models, as explained in Section~\ref{SecIntro}.

The FO and PFSU models
are analytically solvable and
hence are more tractable compared to other economic geography
models. 
Such analytical tractability allows us to analyze the existence and uniqueness of bifurcation points for the full agglomeration.

The preliminary analysis of the MT model showed that the peripheral cities
always emerge next to the center $(i=\pm 1)$.
For this reason, the MT model is not suitable for our purpose of
finding the location of the peripheral cities.
Appendix~\ref{AnalysisMT} gives assumptions of this model.

\subsection{Common ground in economic modeling}

We introduce some general assumptions and 
establish a common ground for the economic geography models.
There are two factors of production (skilled and unskilled workers) and two
sectors (manufacturing, M, and traditional, A). 
Both types of workers consume final goods of two kinds: manufacturing
sector goods and traditional sector goods. 
Workers inelastically supply one unit of each type of labor. 
Skilled workers are mobile between places. 
The number of skilled workers in place $i\in N$ is denoted
by $\lambda_{i}$ under the constraint
 $\sum_{i\in N}\lambda_{i}=1$ in \eqref{ConservationLaw}.
The unskilled workers are immobile and distributed equally
with $L_{i}=L/K$ $({}^\forall i \in N)$.

The transportation costs for M-sector goods
take the iceberg form. That is, for each unit of M-sector goods
transported from place $i$ to place $j\neq i$, only a fraction $1/\tau_{ij}<1$
actually arrives ($\tau_{ii}=1$). It is assumed that 
$\tau_{ij}=\exp(\tau \ m(i,j) \ \tilde{L})$
is a function of a transport cost parameter $\tau>0$, where $m(i,j)$
is an integer expressing the road distance between places $i$ and $j$
and $\tilde{L}$ is the distance unit.
We have $m(i,j)=|i-j|$ for the long narrow economy.
 As the bifurcation parameter,
we introduce the trade freeness (a converse measure of transport costs):
\[
\phi=\exp[-\tau(\sigma-1)\tilde{L}]\in(0,1),
\]
where $\sigma~(>1)$ is a parameter that expresses the constant elasticity of substitution between manufactured varieties.
Then we represent by $\phi_{ij}=\phi^{|i-j|}$
the accessibility between places $i$ and $j$.

\subsection{The FO model}

The FO model has the utility of the form:
\[
u_{i}=\mu\log C_{i}+(1-\mu)\log A_{i}
\]

\noindent
in place $i$,
where $\mu\in(0,1)$ is the share of income spent on manufactures,
$C_{i}=\left[\int_{s\in S}c_{i}(s)^{\frac{\sigma-1}{\sigma}}ds\right]^{\frac{\sigma}{\sigma-1}}$
is the CES composite of manufactures, 
and $c_{i}(s)$ is the consumption
in place $i$ of a variety $s$ of manufactures. 
We assume the no-black-hole condition 
\cite{Forslid.Ottaviano.2003}
\begin{align}\label{no-black-hole} &
\mu < \sigma - 1,
\end{align}
since the violation of this condition denies the existence of bifurcation 
and is empirically unrealistic.

The indirect
utility in place $i$ is given by 
\cite{Akamatsu.et.al.2023}:
\begin{align}\label{IndiUtiFE}
	v_{i}(\boldsymbol{\lambda})=\frac{\mu}{\sigma-1}\log\Delta_{i}(\boldsymbol{\lambda})+\log w_{i}(\boldsymbol{\lambda})+\zeta,
\end{align}

\noindent
where $\Delta_{i}(\boldsymbol{\lambda})=\sum_{j\in N}\phi_{ij}\lambda_{j}$,
$\zeta$ is a constant term,
and $w_i$ is the nominal wage in place $i$
that is given by (see \citet{Gaspar.2020}):
\begin{equation}
w_i(\boldsymbol{\lambda}) =\frac{\mu}{\sigma}\sum_{j \in N}\dfrac{\phi_{ij}\left(1+w_{j}\lambda_{j}\right)}{\sum_{m \in N}\phi_{mj}\lambda_{m}},
	\label{eq:nominal wage FE}
\end{equation}

\noindent
where we use $L_{i}=1$.

For the full agglomeration 
$\bm{\lambda}^{\mathrm{FA}}_0$,
the indirect utilities in places   
$i=0,~\pm j$ 
$(\neq 0)$
are expressed explicitly as
(see Appendix~\ref{The FE model Appendix}) 
\begin{align}
&\hspace{-2mm} v_0 = \log \frac{\hat{\mu}}{1-\hat{\mu}}
	(2k+1), \qquad \hat{\mu} = \frac{\mu}{\sigma} \in (0, 1);
	\label{Payoff_full_Agglomeration_Center}
	\\  \hspace{-2mm}
 & v_j=v_{-j}
	= \log\frac{\hat{\mu}}{1 - \hat{\mu}} +
	\frac{j \mu}{\sigma - 1}\log\phi  
	+ \log
	\left\{
	(\hat{\mu} k + k + 1)\phi^{j} + (1 - \hat{\mu} )
	\left[ (k-j) \phi^{-j} + \sum^{j}_{p=1} \phi^{j-2p} \right] \,
	\right\}. \hspace{-2mm}
	\label{Payoff_full_Agglomeration}
\end{align}

\subsection{The PFSU model}

The PFSU model \cite{Pfluger.Sudekum.2008}
 differs from the FO model in that  
agents have the utility of 
 quasi-linear logarithmic form
and consume housing measured in floor areas, 
which produces a local dispersion force.
The price of housing depends on the number of agents.
The utility is given by
\[
u_{i}=\alpha\log C_{i}+\gamma\log H_{i}+A_{i},
\]

\noindent
where $\alpha$ $(>0)$
 is a preference parameter towards manufactured goods,
 $H_{i}$ stands for the consumption of housing goods,  and $\gamma >0$ is a parameter representing preferences towards consumption of housing. 

 The indirect
utility of a mobile agent is given by \cite{Akamatsu.et.al.2023}:
\begin{align}\label{IndiUtiPFSU}
	v_{i}=
	\frac{\alpha}{\sigma-1}\log\Delta_{i}(\boldsymbol{\lambda}) +
	w_{i}(\boldsymbol{\lambda}) - 
	\gamma \log\frac{\lambda_{i}+1}{h_{i}} + \xi,
\end{align}
where $\xi$ is a constant term and $h_{i}$ denotes the share of housing stock in place $i$. 
The nominal wage $w_i$ is given by
\begin{align}\label{Nominal wage PF}
	w_{i}(\boldsymbol{\lambda})=\frac{\alpha}{\sigma}
	\sum_{j\in N}\dfrac{\phi_{ij}\left(1+\lambda_{j}\right)}
	{\sum_{m\in N}\phi_{mj}\lambda_{m}},
\end{align}
\noindent
where we use $L_i=1$.

The indirect utilities 
for $\bm{\lambda} = \bm{\lambda}^{\mathrm{FA}}_0$
 are given by (see Appendix~\ref{The PFSU model Appendix})
\begin{align}
	& v_{0} =\dfrac{\alpha}{\sigma}\left(2k+2\right)-\gamma\log2 + \xi,
	\label{eq:indirectutilitiesPFSUmodel0} \\
	& 
v_j=v_{-j}=\dfrac{\alpha}{\sigma}
	\left[ \phi^{j}(k+2)+\phi^{-j}(k-j)+\frac{\phi^{j}-\phi^{-j}}{\phi^{2}-1}\right]
	 +\dfrac{\alpha j}{\sigma-1}\log\phi + \xi 
	 \quad (1\le j \le k).
	 \label{eq:indirectutilitiesPFSUmodel}
\end{align}


\section{Bifurcation mechanism of economic geography models}\label{Stability of full agglomeration}

We elucidate how peripheral
cities emerge for the three economic geography models
(Section~\ref{Economic geography models}).
Although these models display different kinds of bifurcation behaviors,
the bifurcation mechanism 
presented in Section~\ref{Bifurcation theory full}
synthetically describe these behaviors.

\subsection{Stability of full agglomerations at several places}\label{Other than center}

To demonstrate
that the center is the best place for the full agglomeration
regarding stability,
we conduct stability analysis
of $K=7$ places
for the three models 
with typical values of economic parameters.
The red solid line in Fig.~\ref{SustainabilityMONO}
shows the range of $\phi \in (0, 1)$ in which
the full agglomeration $\bm{\lambda}^{\rm FA}_j$
$(0 \le j \le 3)$ is stable.
Recall that 
$\bm{\lambda}^{\rm FA}_{j}$ 
is an invariant pattern (Proposition~\ref{Invariant proposition}).

For the FO
and PFSU models (Figs.~\ref{SustainabilityMONO}(a) and (b)),
the full agglomeration 
$\bm{\lambda}^{\rm FA}_0$ at the center has the 
widest range of stable state $\phi\in(\phi^{\rm s},1)$
and is the one that becomes stable (sustainable) first among 
full agglomerations when the trade freeness increases from a low value.
For the MT model,
the full agglomeration at the center is
equally predominant as the full agglomerations elsewhere.

\begin{figure}[!b]
\begin{tiny}
\begin{center}
\begin{tabular}{c@{\hspace{-4mm}}c@{\hspace{-4mm}}c}
\includegraphics[width = 50mm]{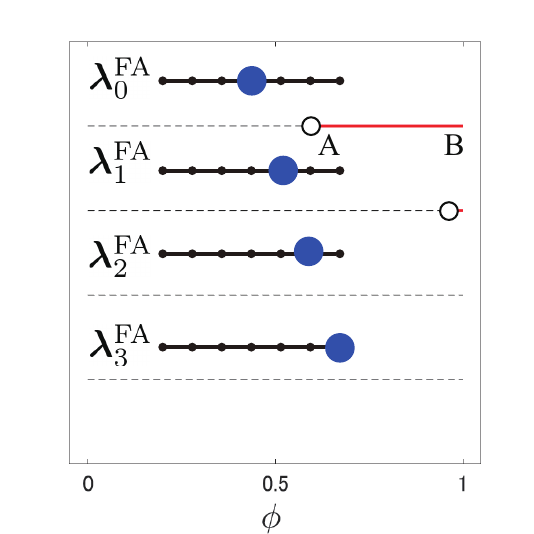} &
\includegraphics[width = 50mm]{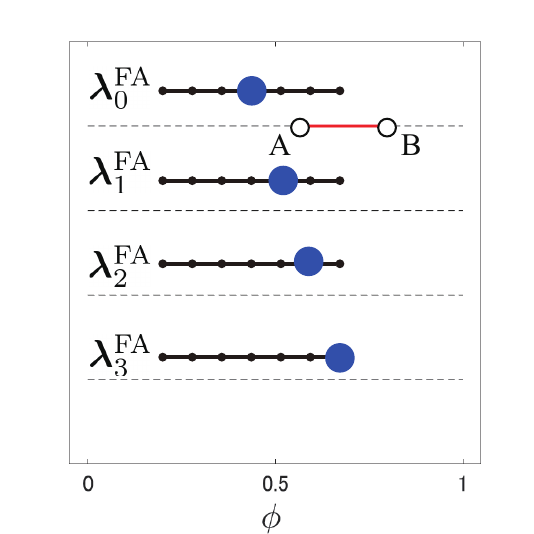} &
\includegraphics[width = 50mm]{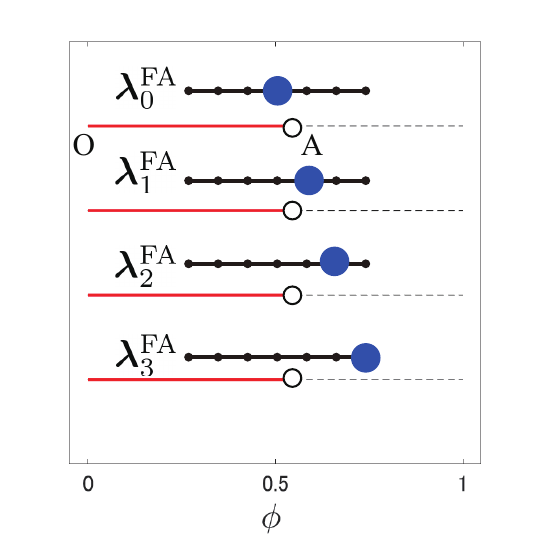}
\\
(a) The FO model & (b) The PFSU model &
(c) The MT model
\end{tabular}
\end{center}
\end{tiny}
\caption{
The range of $\phi$ of stable full agglomerations
$\bm{\lambda}=\bm{\lambda}^{\rm FA}_j$
($\sigma = 6.0$
for all the
three models;
$\mu = 0.4$
that satisfies the no-black-hole condition ($\mu<\sigma-1$)
for the FO model;
$(\alpha, \gamma) = (0.8, 0.1)$
for the PFSU model;
$\theta=0.2$ for the MT model
\cite{Murata.Thisse.2005};
red solid line: stable; broken line: unstable;
$\bigcirc$: sustain point)}
\label{SustainabilityMONO}
\end{figure}
\FloatBarrier

We, accordingly, propose the following conjecture and
examine the bifurcation of $\bm{\lambda}^{\rm FA}_0$ specifically.

\begin{conjecture}
	The center is the best place for the full agglomeration
	regarding stability.
\end{conjecture}

\subsection{Bifurcation at a sustain point}\label{Bifurcation from a full}

We conduct a numerical bifurcation analysis of $K = 7$ places
for the three economic geography models.
Figure~\ref{BifAnaPFFOPFSU} shows the paths of full agglomeration
$\bm{\lambda}^{\rm FA}_0$ at the center
and bifurcating paths that branch at the sustain points shown by $(\circ)$.
The vertical axis is the population $\lambda_0$
at the central place $(i=0)$
and the horizontal axis is the trade freeness $\phi \in (0,1)$.
Solid curves express stable solutions,
while broken curves denote unstable ones.

\begin{figure}[!t]
\begin{small}
\begin{center}
\begin{tabular}{ccc}
\includegraphics[scale=0.6]{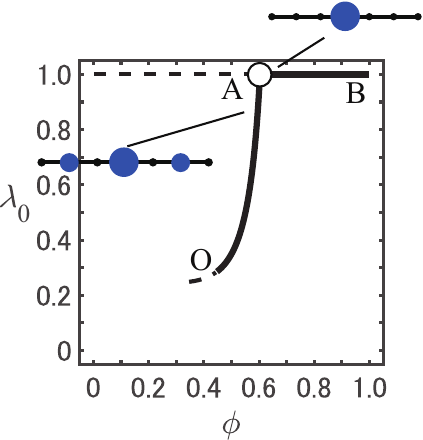} &
\includegraphics[scale=0.6]{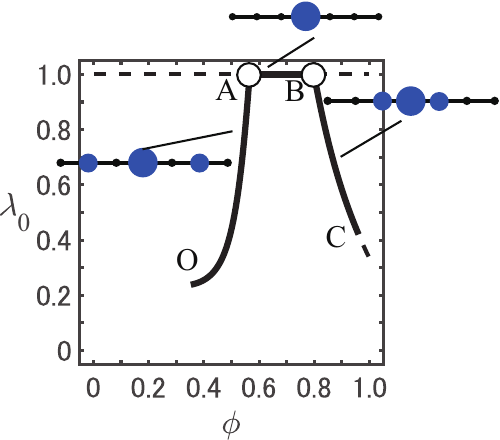} &
\includegraphics[scale=0.6]{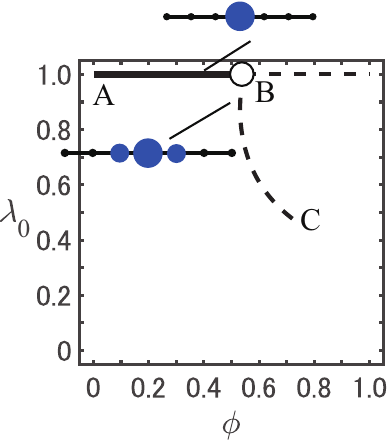} \\
(a) The FO model & (b) The PFSU model &
(c) The MT model
\end{tabular}
\end{center}\end{small}
\caption{
Bifurcating paths
emanating from the state of the full agglomeration at the center
(parameter values employed are given in
the caption of Fig.~\ref{SustainabilityMONO};
solid line: stable; broken line: unstable;
$\bigcirc$: sustain point)}
\label{BifAnaPFFOPFSU}
\end{figure}
\FloatBarrier

First, we observe the sustain point(s) in the stable full agglomeration
(during AB).
\begin{itemize}
\item
For the FO model,
the full agglomeration
has the unique sustain point $\phi^{\rm s}$ at point A,
and is stable for $\phi\in(\phi^{\rm s},1)$
and unstable for $\phi\in(0,\phi^{\rm s})$.
\item
For the PFSU model,
the full agglomeration has two sustain points
$\phi^{\rm s1}$ and $\phi^{\rm s2}$
(at A and B, respectively),
stable for $\phi\in(\phi^{\rm s1},\phi^{\rm s2})$
and unstable for
$\phi\in\left\{ (0,\phi^{\rm s1})\cup(\phi^{\rm s2},1)\right\}$.
\item
For the MT model,
as $\phi$ increases,
the stable full agglomeration
becomes unstable at the sustain point B.
\end{itemize}

By a bifurcation analysis at each sustain point, 
we find two kinds of bifurcating paths:
one with a single peripheral city and another with two peripheral cities
(cf., Proposition~\ref{BifurCurveV_2}).
The bifurcating paths with one peripheral city are unstable for all the models;
accordingly, only the curves
with two peripheral cities (OA and BC)
are included in Fig.~\ref{BifAnaPFFOPFSU}.
There is at most one stable bifurcating path
for each sustain point (cf., Proposition~\ref{Stability_Bifcurve}).

For the FO and PFSU models,
each sustain point has one stable bifurcating path
just after bifurcation.
As $\phi$ increases from a small value,
the stable bifurcating path OA
transits to stable full agglomeration
$\bm{\lambda}^{\rm FA}_0$ (AB).
There is a continuation of stable paths
(cf., Proposition~\ref{Stability_of_Bifcurve}).
The bifurcated path OA displays a core--periphery pattern
with a large central city surrounded by two peripheral cities
located two steps away from the center
($\delta_{\rm s}=2$).
The full agglomeration emerges at the center
by steadily absorbing
and finally nullifying the (mobile) population of peripheral cities.
The PFSU model encounters
peripheral city formation both for sufficiently high and low
trade freeness.
For the MT model, the bifurcating path branches backward
($\phi<\phi^{\rm s}$)
and is unstable just after bifurcation
(cf., Proposition~\ref{Stability_of_Bifcurve}).

\subsection{Parameter dependency of
the location of peripheral places}%
\label{Parameter dependence FE}

We numerically obtain the distance $\delta_{{\rm s}}$
 from the center
of the twin peripheral cities
for the FO model with $K=7$ places.
Figure~\ref{SustainBifFOPFPFSU}(a) shows the contours of $\delta_{{\rm s}}$
in the range $(0,1)\times (0,1)$ of the space of $(1/\sigma,\mu)$
for model parameters.
As $1/\sigma$ and $\mu$ increase, $\delta_{{\rm s}}$
enlarges one by one from $\delta_{{\rm s}}=1$.
Smaller $\sigma$ enlarges economies of scale,
and larger $\mu$ increases the manufacturing sector.
An increase of $\delta_{{\rm s}}$
creates an \textit{agglomeration shadow} [Arthur, 1990; Ikeda et al., Fig.~5, 2017].
By contrast, for large $\sigma$ and small $\mu$,
we have $\delta_{{\rm s}}=1$; the peripheral cities emerge 
next to the central city,
and a discrete version of a hump-shaped megalopolis emerges.

Thus, we have observed the dependence of agglomeration patterns
on the values of economic parameters.
This dependence possibly is a source of the diversity of
the population distribution of a chain of cities worldwide
and may help
the explanation of the historical emergence of megalopolises.
This result is in line with
the analytical results on the stability
of the single agglomeration in the pioneering
work
that considered a similar
economic geography model with a one-dimensional unbounded continuous location space 
\cite{Fujita.Krugman.1995}.

The PFSU model also has parameter dependency.
Figure~\ref{SustainBifFOPFPFSU}(b)
shows the contours of $\delta_{\rm s}$ in the range $(0,1)\times (0,1)$
in $(\alpha,\gamma)$-space for $\sigma=6.0$.
While $\delta_{\rm s} = 1$ and 2
are attainable at the lower sustain point A,
only $\delta_{\rm s} = 1$ at the upper one B
in Fig.~\ref{BifAnaPFFOPFSU}(b).
For the lower sustain point,
higher $\alpha$
 and lower $\gamma$
 increases $\delta_{\rm s}$.

The MT model does not have parameter dependency
as $\delta_{\rm s}=1$ holds always
(cf., Appendix~\ref{The MT model Appendix}).

\begin{figure}[!htbp]
\begin{center}
\begin{tabular}{@{\hspace{-3mm}}c@{\hspace{-3mm}}c@{\hspace{-3mm}}c}
\includegraphics[width = 5.4cm]{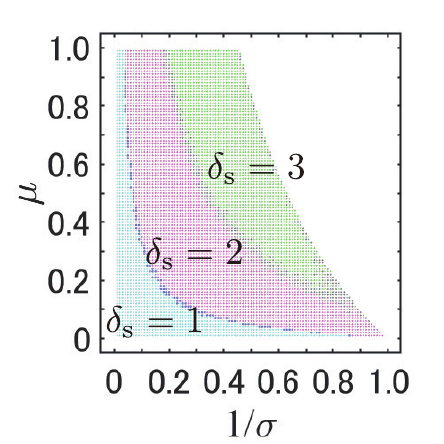} &
\includegraphics[width = 5.4cm]{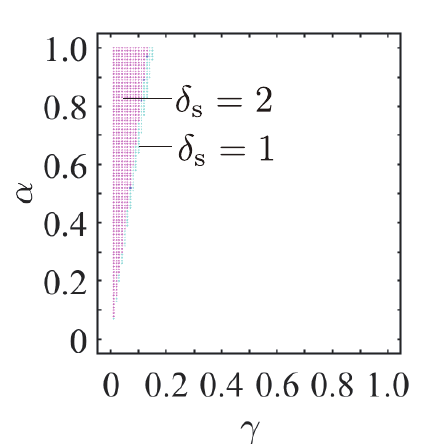} &
\includegraphics[width = 5.4cm]{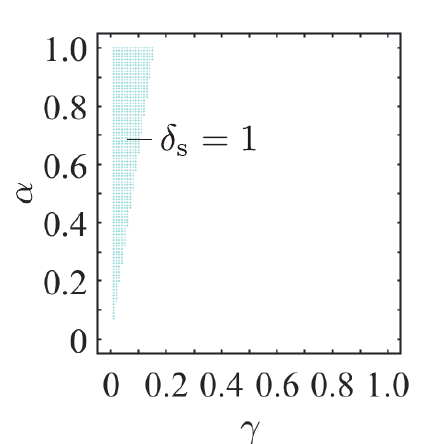} 
\\
& The lower sustain point & The higher sustain point
\\
\noalign{\vskip 0.5ex}
(a) The FO model & \multicolumn{2}{c}{(b) The PFSU model ($\sigma=6.0$)}
\end{tabular}
\end{center}
\caption{Contour maps of $\delta_{\rm s}$
for $K = 7$ places
drawn on parameter spaces
(sky blue area: $\delta_{\rm s} = 1$;
pink: $\delta_{\rm s} = 2$;
green: $\delta_{\rm s} = 3$)}
\label{SustainBifFOPFPFSU}
\end{figure}
\FloatBarrier


\section{Bifurcation analysis for the FO model}\label{AnalysisGeography}

We conduct a further bifurcation analysis for 
the FO model with several numbers of places.
The analysis for $K=5$ places  
is used in the study of the historical change in the
population distribution
of Japan's Main Island in Section~\ref{Use of EG analysis}.

\subsection{Stability of full agglomerations and twin cities}


For the FO model,
we investigate the stability of  
full agglomerations  
$\bm{\lambda}^{\rm FA}_j$
and twin cities 
$\bm{\lambda}^{\rm Twin}_j$. 
These are solutions for the governing equation for any $\phi$
 (Proposition~\ref{Invariant proposition}).
They are stable in
the ranges of $\phi$ plotted by
the red solid lines in Fig.~\ref{Sustainability of K=5-15}
for several numbers of places
($K=5,7,11,15$).

\begin{figure}[!t]
\begin{small}
  \begin{center}
    \begin{tabular}{@{\hspace{-2mm}}c@{\hspace{-2mm}}c}
      \includegraphics[scale=0.135]{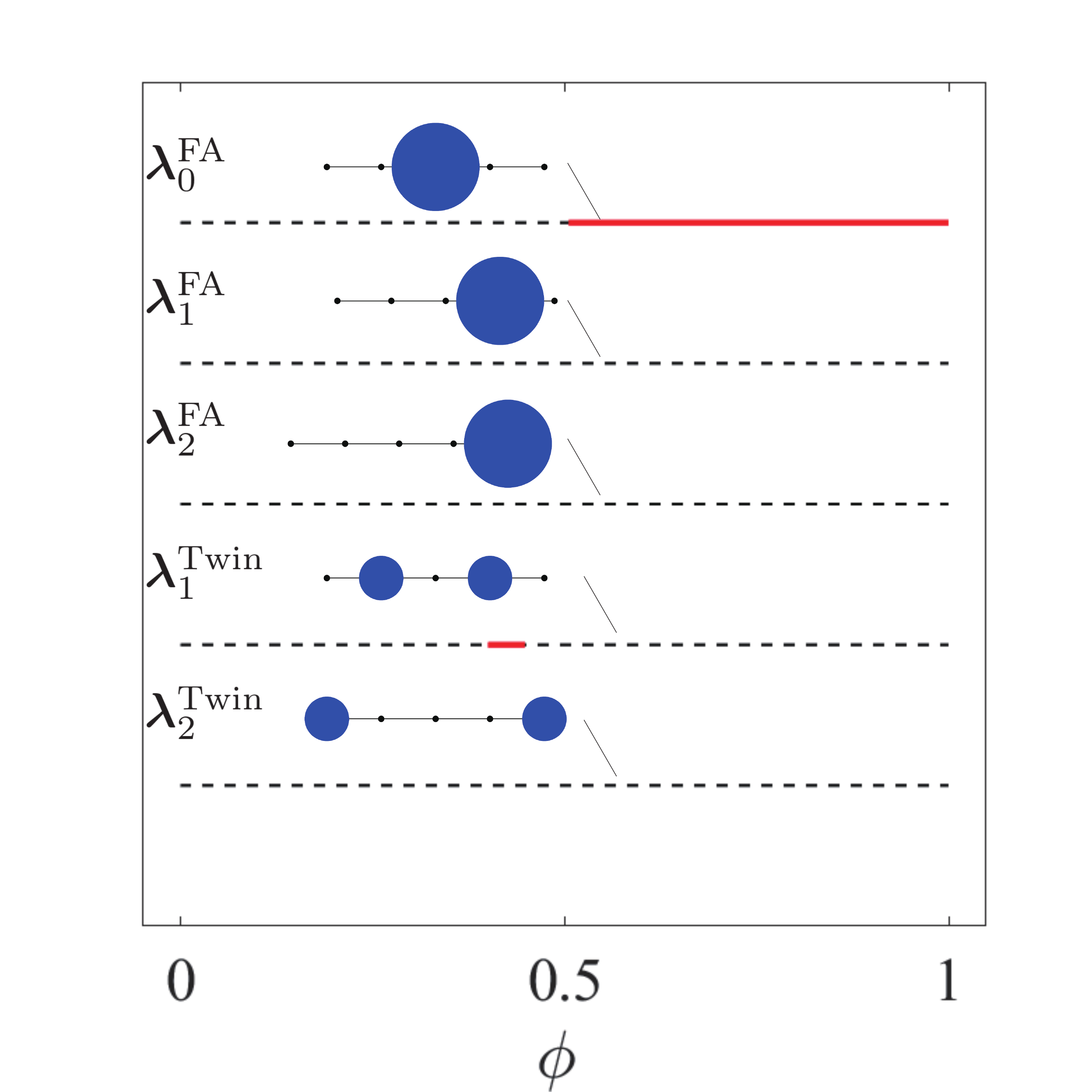} &
      \includegraphics[scale=0.135]{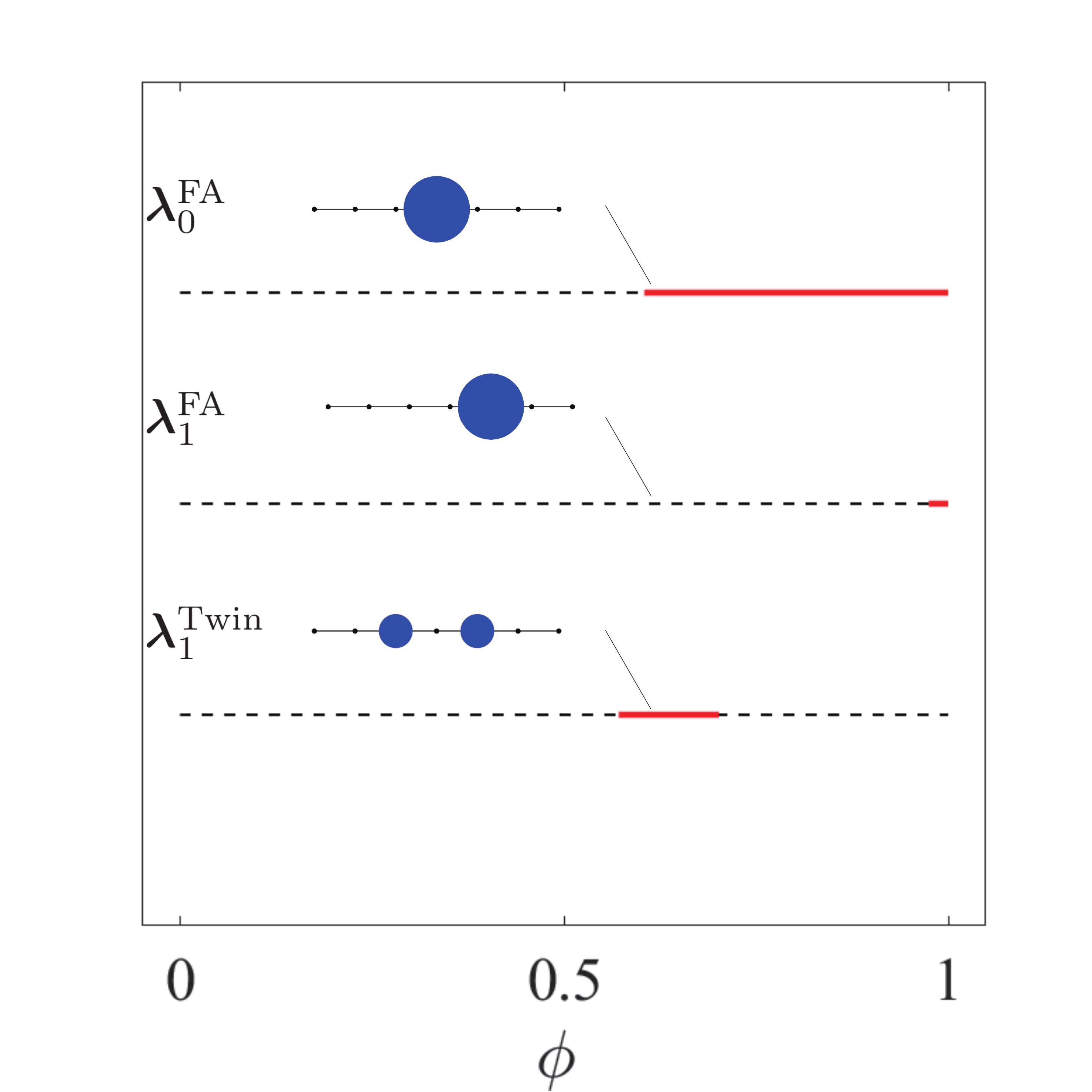} \\
(a) $K=5$ & (b) $K=7$ \\
      \includegraphics[scale=0.135]{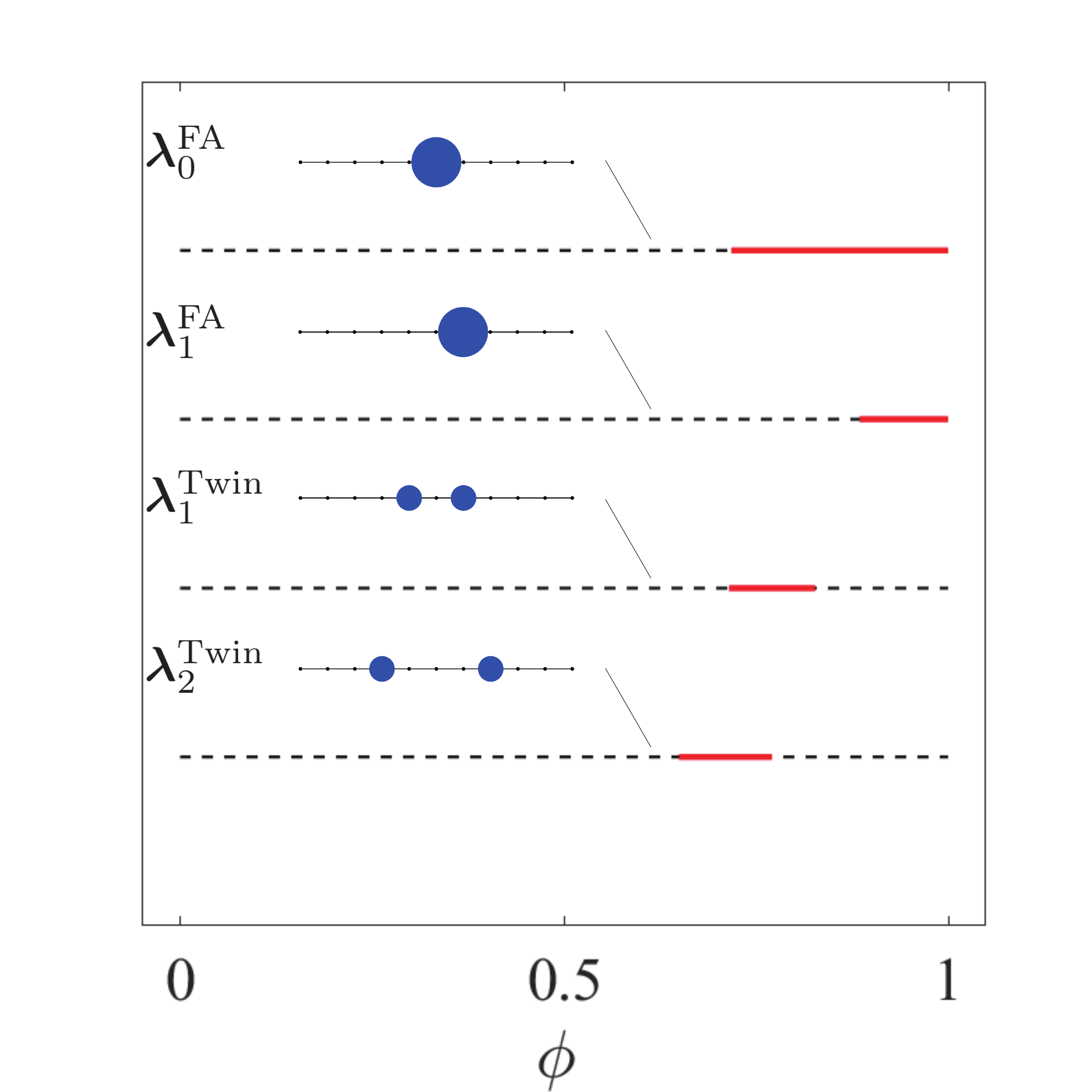} &
      \includegraphics[scale=0.135]{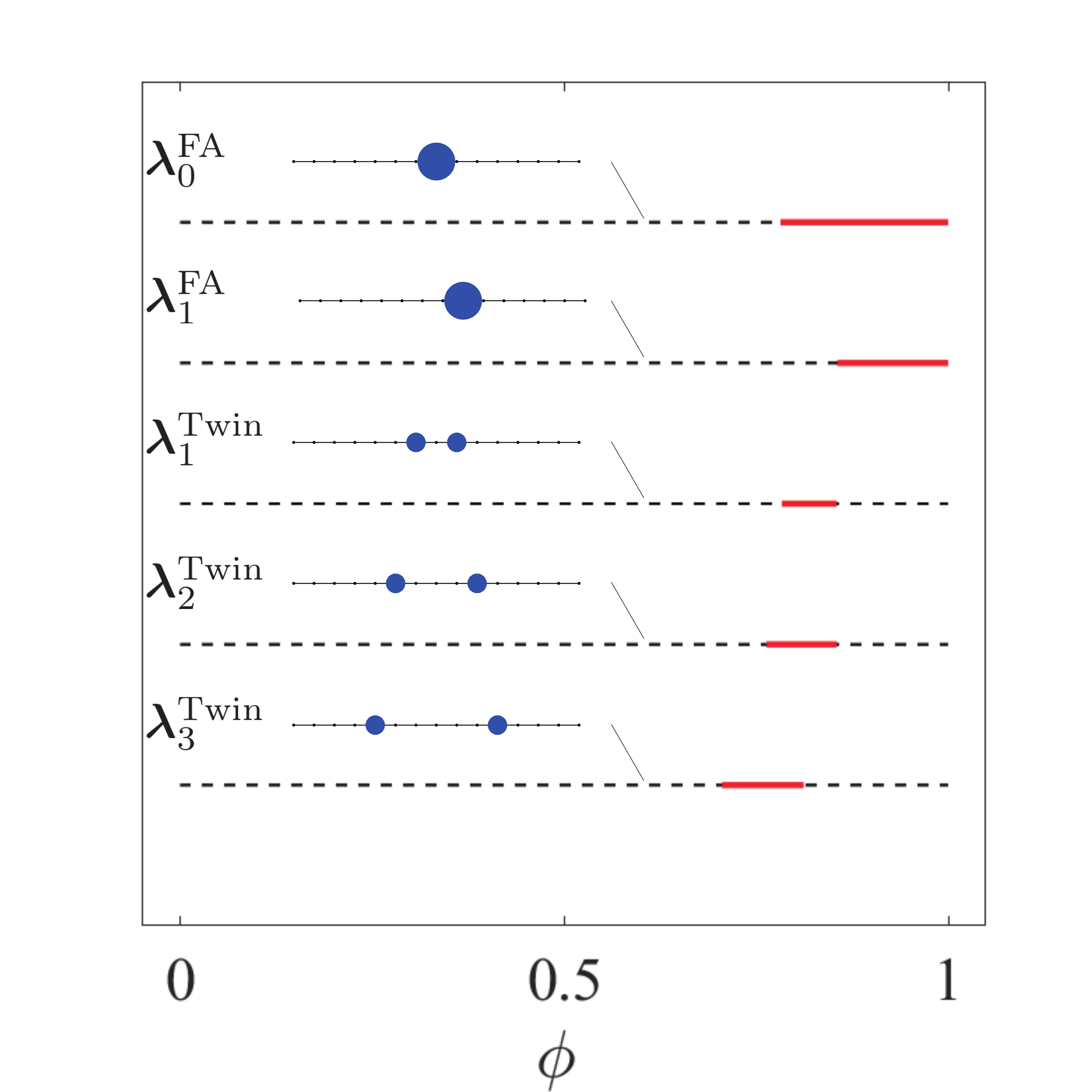} \\
(c) $K=11$ & (d) $K=15$ 
    \end{tabular}
  \end{center}
  \end{small}
  	\caption{
The range of $\phi$ of stable full agglomerations
and twin cities for the FO model for $(\sigma,\mu)=(6.0,0.4)$
and $\tilde{L}=1/6$
(red solid line: stable; broken line: unstable)}
\label{Sustainability of K=5-15}
\end{figure}
\FloatBarrier

Similarly to the result in Fig.~\ref{SustainabilityMONO} for the three models,
the full agglomeration $\bm{\lambda}^{\rm FA}_0$ at the center
has superior stability to full agglomerations elsewhere  
and becomes stable first when the trade freeness 
$\phi$ increases from a low value.
For this reason, we focus on 
$\bm{\lambda}^{\rm FA}_0$ from now on.

For a relatively small number of places 
($K=5, 7$ shown in Figs.~\ref{Sustainability of K=5-15}(a) and (b)),  
among the twin cities, 
only $\bm{\lambda}^{\rm Twin}_{1}$
located next to the center is stable  
(near $\phi=0.44$ for $K=5$ and $\phi=0.65$ for $K=7$).
For more cities $(K=11,15$ in panels (c) and (d)),
stable twin cities tend to be located outwards
($\delta_{\rm s}=1,~2$ for $K=11$ 
and $\delta_{\rm s}=1,~2,~3$ for $K=15$).
The twin cities are stable 
for intermediate values of $\phi$,
while the full agglomeration $\bm{\lambda}^{\rm FA}_0$
is stable
for large values of $\phi$.
It implies an inevitable transition from the twin cities to the full agglomeration
as $\phi$ increases from an intermediate to a large value, 
as will be observed in the numerical analysis of 
Section~\ref{AgglomerationBehaviorsMore}.

Figure~\ref{Sustainability of K=5-15} demonstrates 
the coexistence of multiple stable solutions
for some range of $\phi$.
For example, $\lambda_0^{\mathrm{FA}}$,
$\lambda_1^{\mathrm{Twin}}$, and $\lambda_2^{\mathrm{Twin}}$
are stable 
for $\phi \in [0.717, 0.770]$ for $K=11$ in panel (c).
The number of multiple stable solutions
increases from two, three to four
as $K$ increases from 7, 11, to 15.

\subsection{Sustain point of full agglomeration
at the center}\label{Sustainability NEG}

The full agglomeration
$\bm{\lambda}^{\rm FA}_0$ at the center of
the FO model
has the local sustain point 
$\phi=\phi^{\rm s}_\delta 
\in(0,1)$ ($v_\delta-\bar{v}=0$ for some $\delta~( \neq 0)$)
and the sustain point 
$\phi=\phi^{\rm s}\in(0,1)$  
($v_{\delta_{\rm s}}-\bar{v}=0$ for some $\delta_{\rm s}~(\neq 0)$
and $v_i-\bar{v}\le 0$ for all $i \neq \delta_{\rm s}$).

	\begin{lemma}\label{Sustainable_Full_Agglomeration_NEW}
For 
$\bm{\lambda}^{\rm FA}_0$ of
the FO model under the no-black hole condition \eqref{no-black-hole},
we have:

\noindent
		(i) A local sustain point $\phi=\phi^{\rm s}_\delta \in(0,1)$  
exists for every $\delta \neq 0$. 

\noindent
(ii) The local sustain point is unique
 for each $\delta\in\{1,\ldots,6\}$.
			\end{lemma}
	\begin{proof}
		See Appendix~\ref{The FE model Appendix} for the proof.
			\end{proof}

	\begin{proposition}\label{Stability_Full_Agglomeration}
For $\bm{\lambda}^{\rm FA}_0$ of
the FO model under the no-black hole condition \eqref{no-black-hole},
we have:

\noindent
(i) A sustain point 
		 $\phi^{\rm s}=
		 \max_{j \neq 0} \, \phi_\delta^{\rm s}\in(0,1)$
 exists and $\bm{\lambda}^{\rm FA}_0$
 is stable for 
 $\phi\in(\phi^{\rm s},1)$.
 
 \noindent
 (ii) 
 The sustain point is unique and $\bm{\lambda}^{\rm FA}_0$
 is unstable for $\phi\in(0,\phi^{\rm s})$
 for each $\delta_{\rm s}\in\{1,\ldots,6\}$.
	\end{proposition}
	\begin{proof}
		See Appendix~\ref{The FE model Appendix} 
		for the proof.
			\end{proof}

\subsection{Agglomeration behaviors for several numbers of places}\label{AgglomerationBehaviorsMore}

\begin{figure}[!p]
\begin{small}
  \begin{minipage}{.575\textwidth}
  \begin{center}
      \includegraphics[scale=0.5]{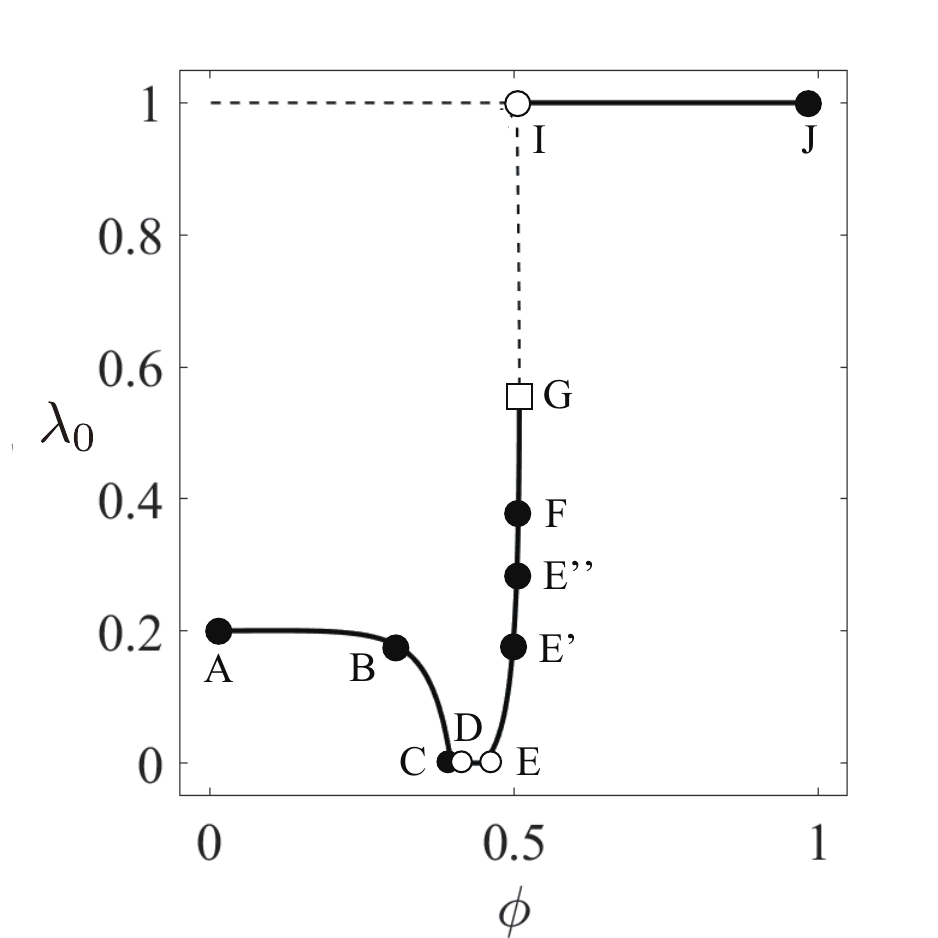} 
   \end{center}
   \end{minipage}
  \begin{minipage}{.33\textwidth}
  \begin{center}
  \begin{tabular}{l@{\hspace{-3mm}}l}
      \includegraphics[scale=0.1]{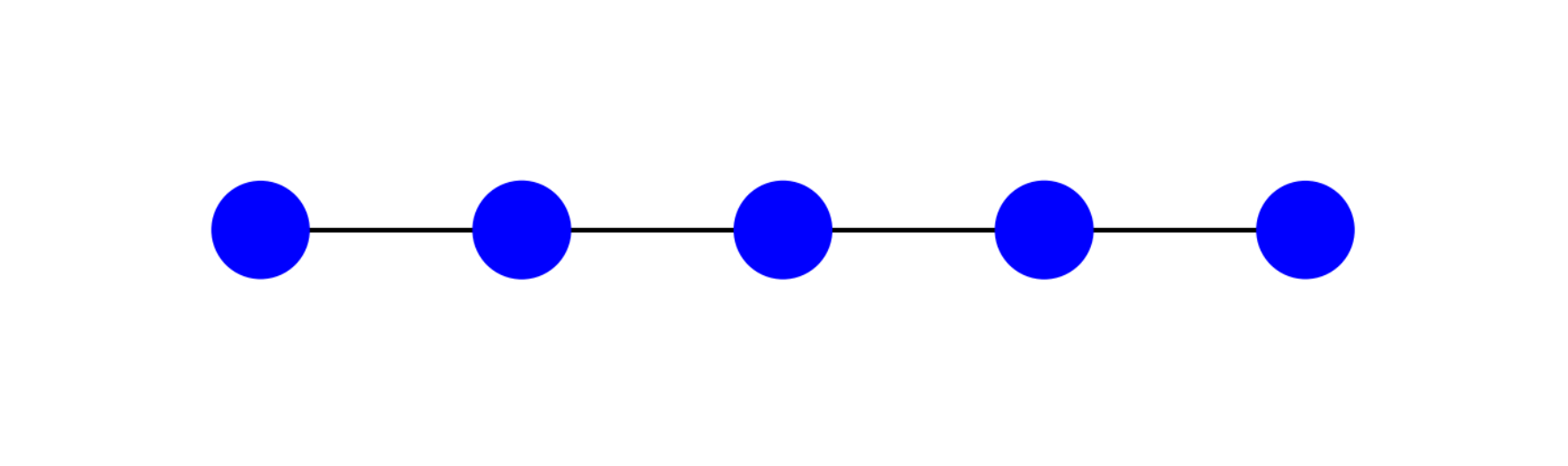}
       & \raisebox{3.3mm}{A} \\
       \noalign{\vskip -2ex}
      \includegraphics[scale=0.1]{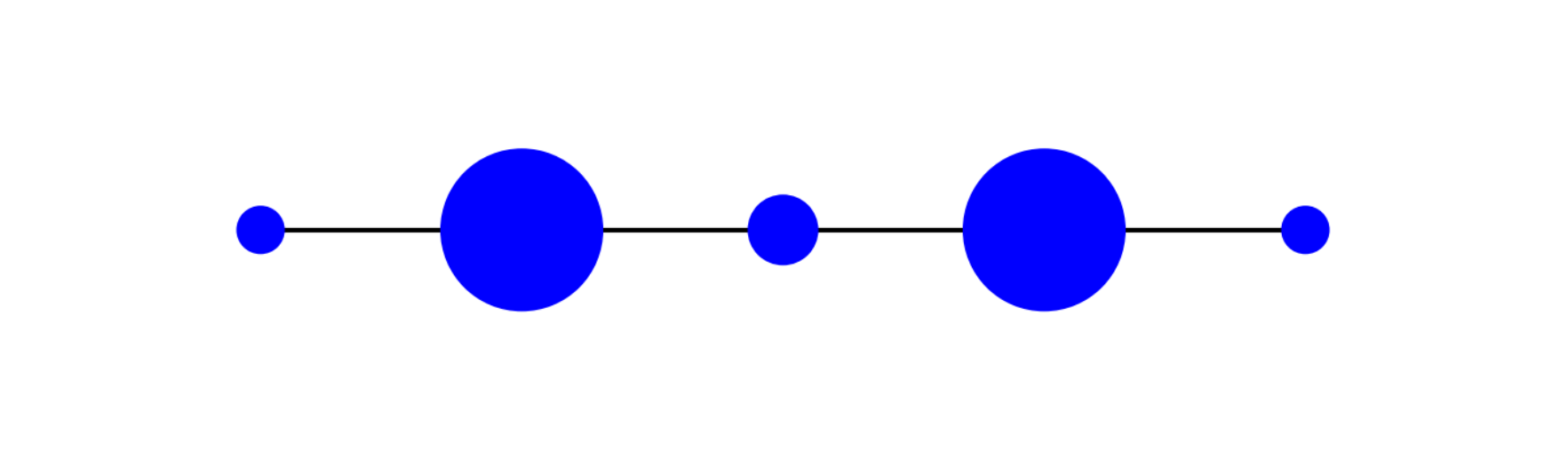} 
      & \raisebox{3.3mm}{B} \\
       \noalign{\vskip -2ex}
      \includegraphics[scale=0.1]{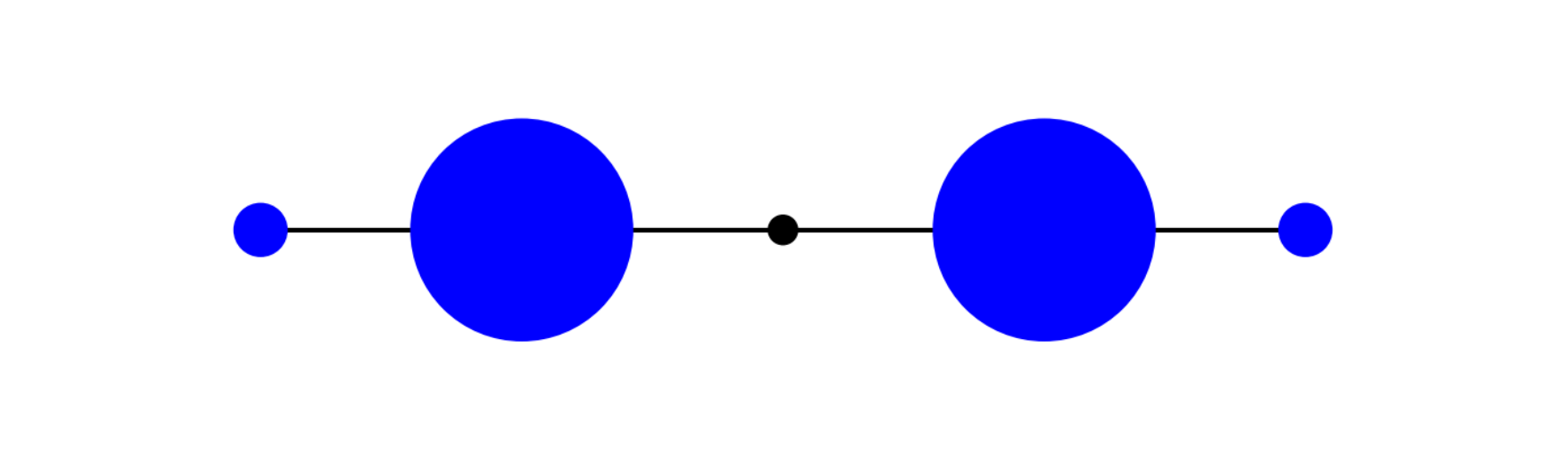}
      & \raisebox{3.3mm}{C} \\
       \noalign{\vskip -2ex}
      \includegraphics[scale=0.1]{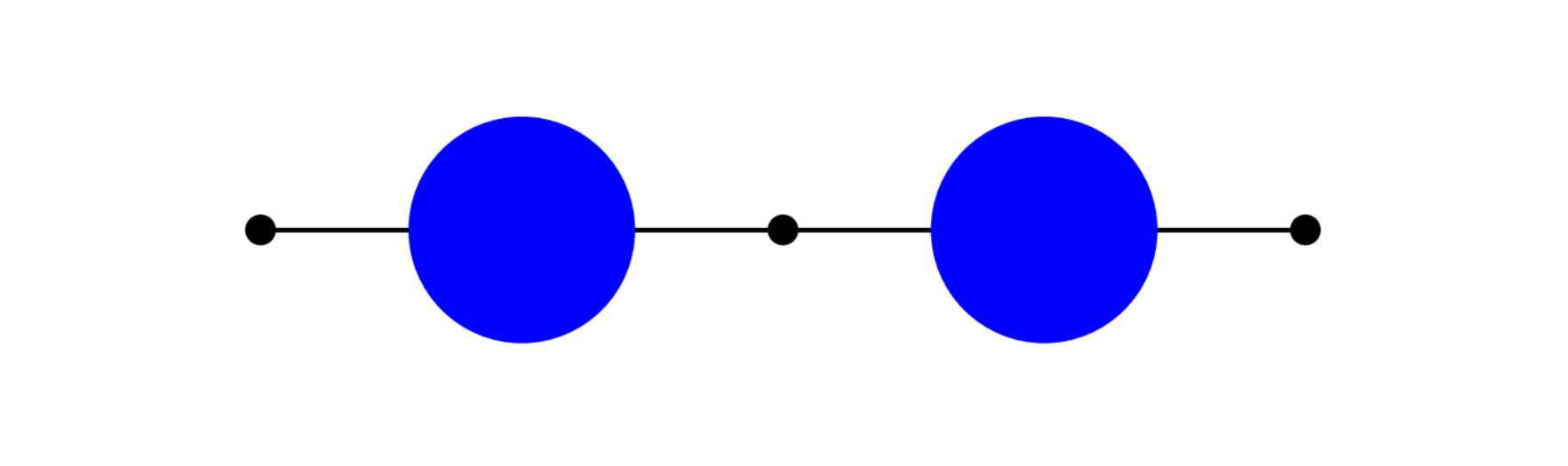}
      & \raisebox{3.3mm}{DE} \\
	  \noalign{\vskip -2ex}
	  \includegraphics[scale=0.1]{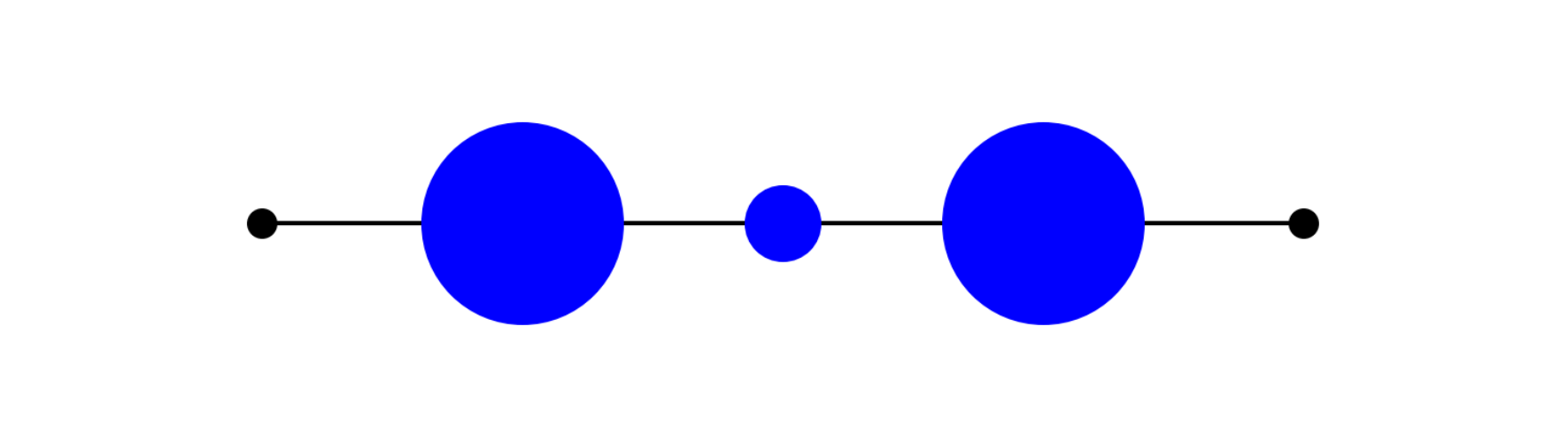}
      & \raisebox{3.3mm}{E'} \\
	  \noalign{\vskip -2ex}
	  \includegraphics[scale=0.1]{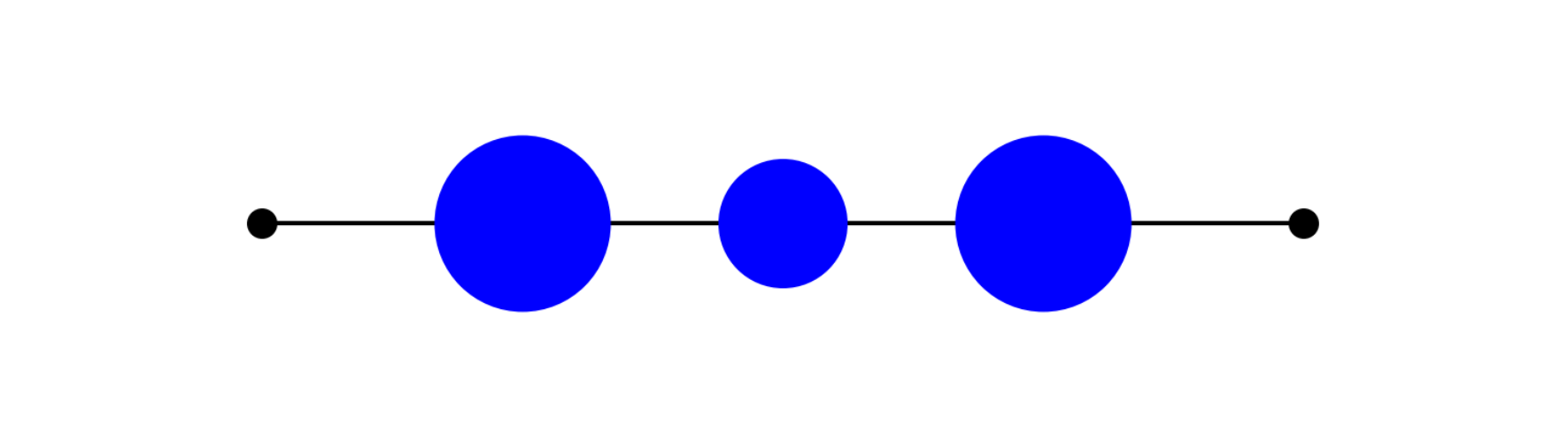}
      & \raisebox{3.3mm}{E''} \\
      \end{tabular} 
      
      \vspace{-2mm}
      Dawning

      \vspace{2mm}
  \begin{tabular}{l@{\hspace{-3mm}}l}
      \includegraphics[scale=0.1]{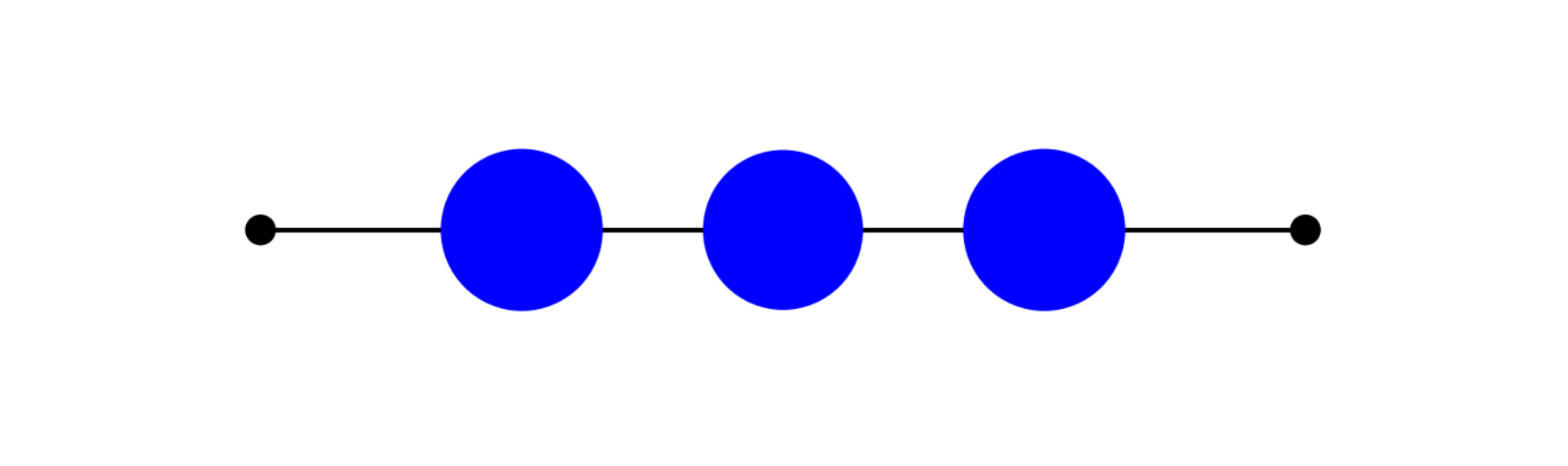}
       & \raisebox{3.3mm}{F} \\
       \noalign{\vskip -1ex}
      \includegraphics[scale=0.1]{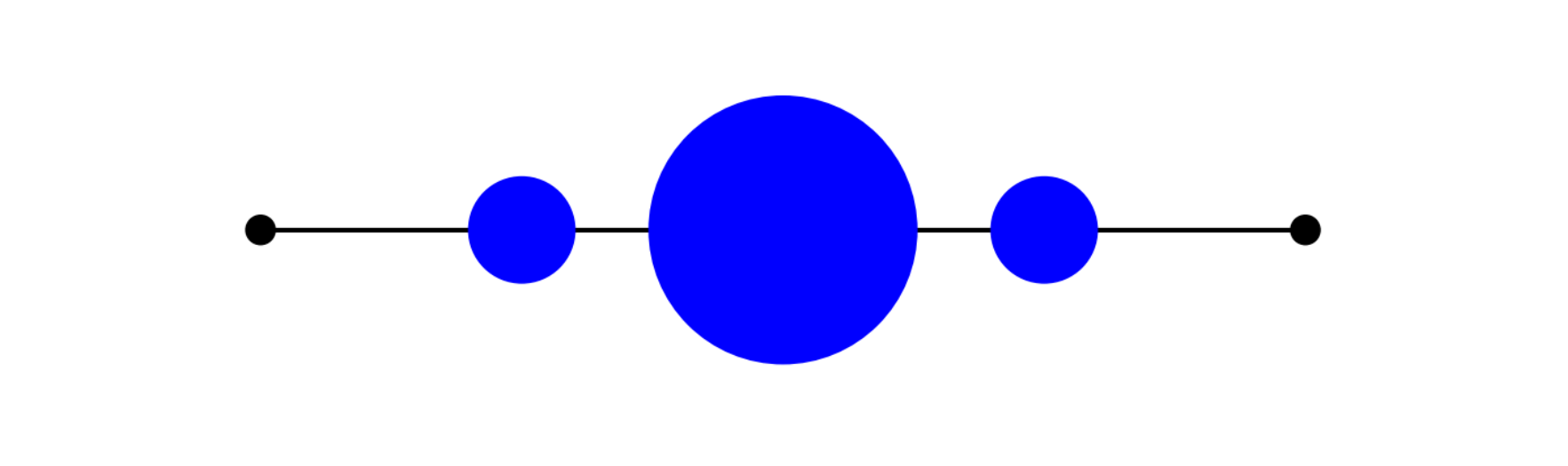} 
      & \raisebox{3.3mm}{G} \\
      \end{tabular} 

      \vspace{-2mm}
      Core--periphery 

      \vspace{2.5mm}
  \begin{tabular}{l@{\hspace{-3mm}}l}
      \includegraphics[scale=0.1]{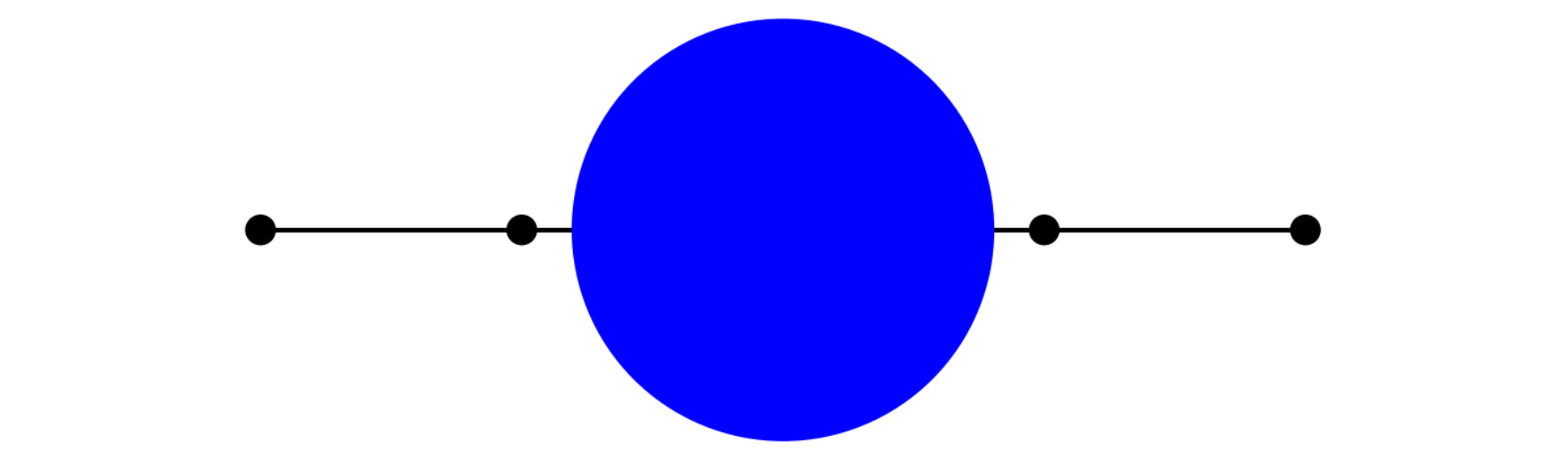}
       & \raisebox{3.3mm}{IJ} \\
      \end{tabular} 

      Full agglomeration
  \end{center}
   \end{minipage}
\end{small}
	\vspace{2mm}
    	\caption{Solution curves 
    	of $K=5$ places
    	for the FO model with
  $(\sigma, \mu) = (6.0, 0.4)$
  (solid line: stable; broken line: unstable; 
$\circ$: sustain point;
$\square$: local maximum point of $\phi$; 
$\bullet$: reference point)}
	\label{K=5,sigma=6}

	\vspace{2mm}
\begin{small}
 \begin{minipage}{.58\textwidth}
  \begin{center}
      \includegraphics[scale=0.16]{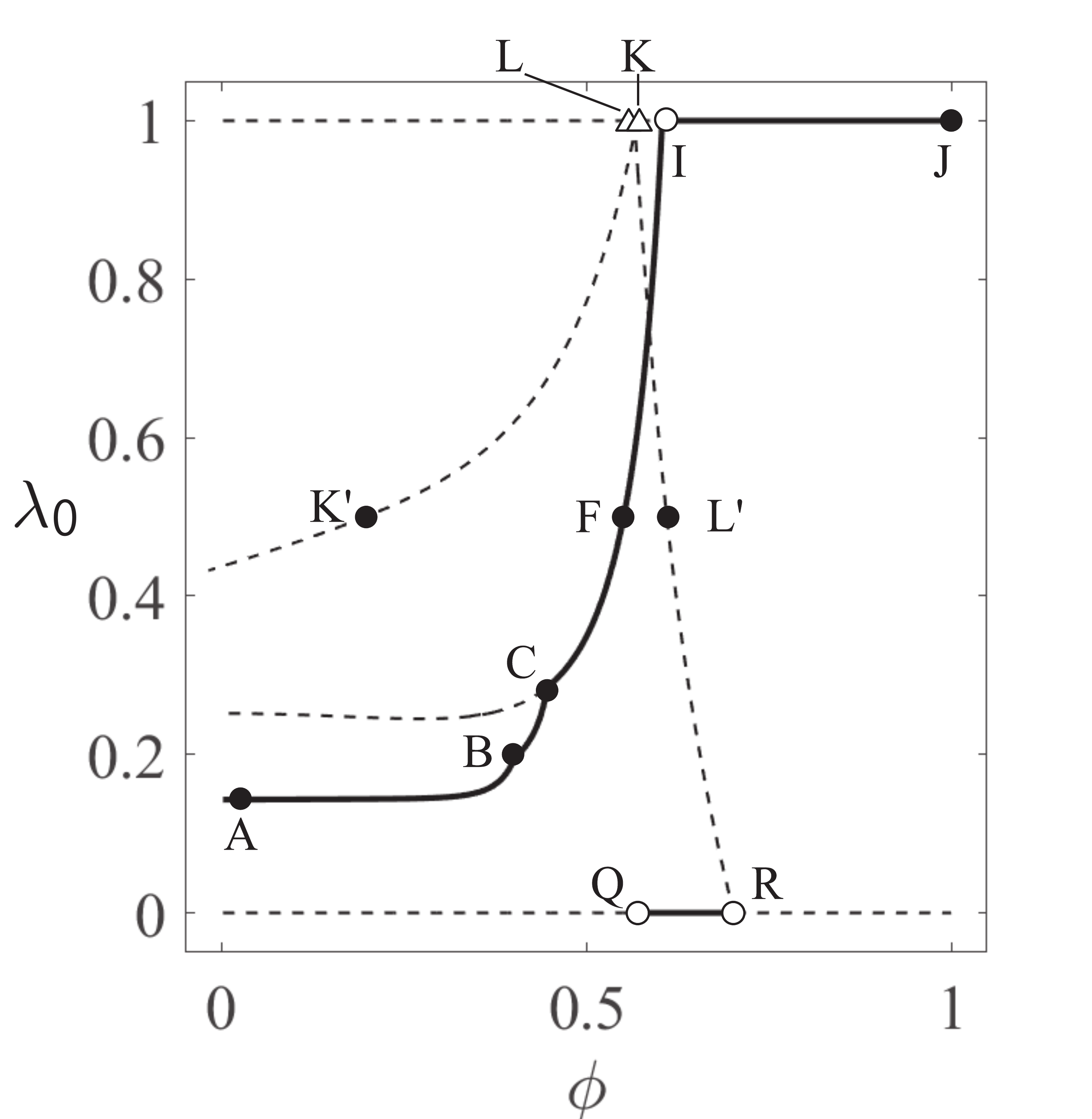} 
   \end{center}
   \end{minipage}
   \hspace{-3mm}
  \begin{minipage}{.32\textwidth}
    \begin{center}
  \begin{tabular}{c@{\hspace{-1mm}}l}
      \includegraphics[scale=0.08]{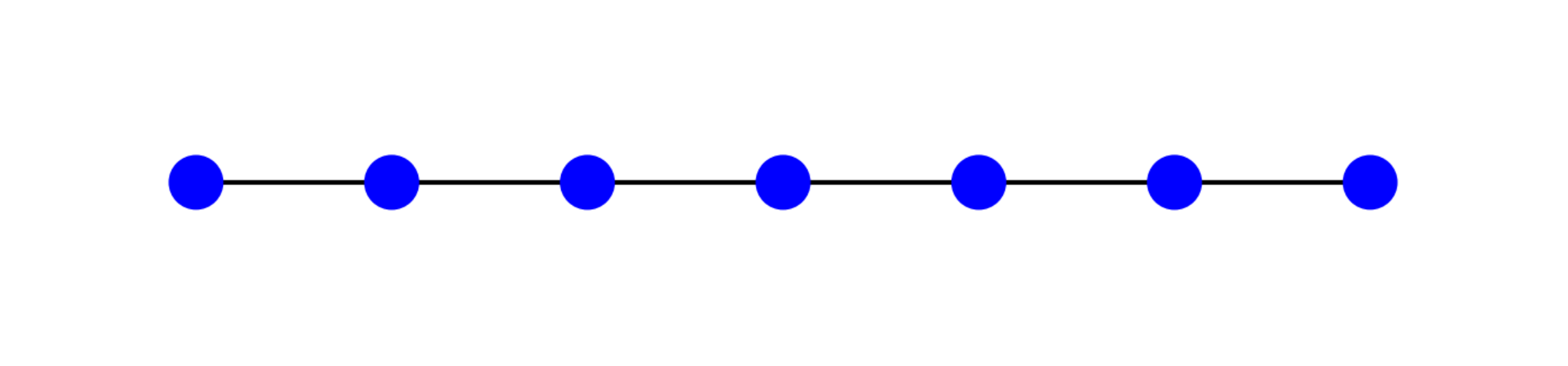}
       & \raisebox{3.1mm}{A} \\
       \noalign{\vskip -2ex}
      \includegraphics[scale=0.08]{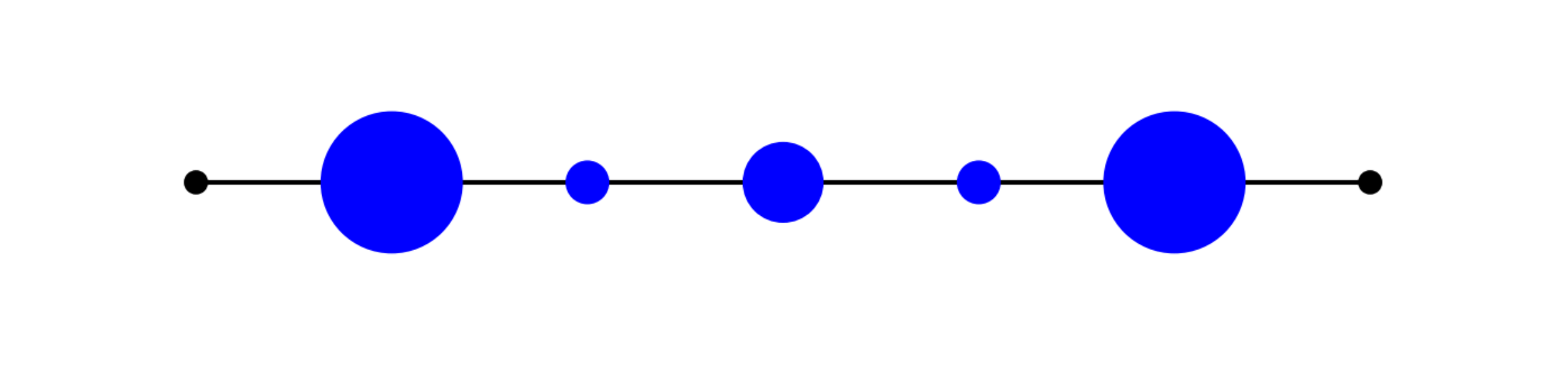} 
      & \raisebox{3.1mm}{B} \\
      \noalign{\vskip -2ex}
      \includegraphics[scale=0.08]{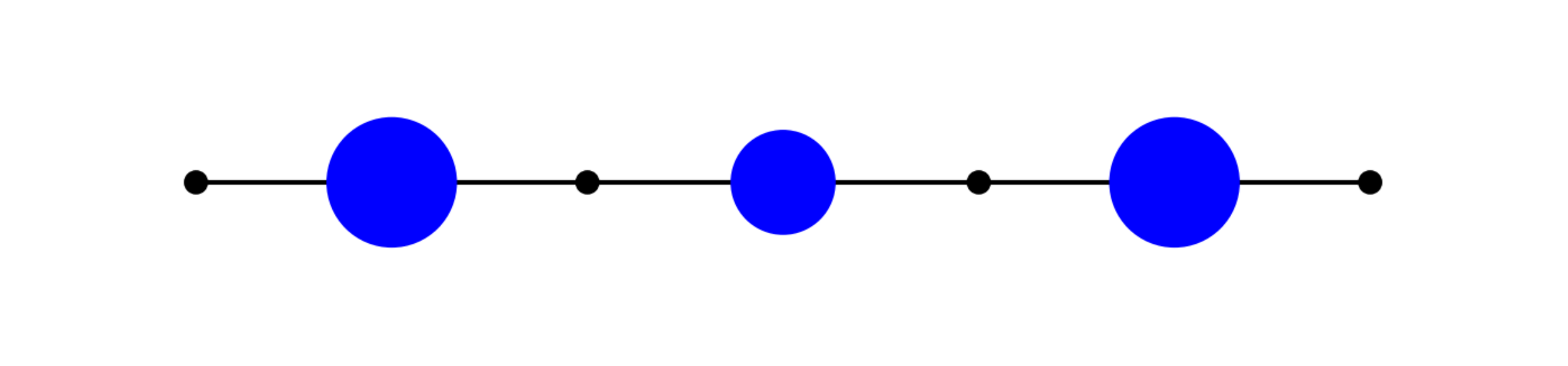}
      & \raisebox{3.1mm}{C} \\
     \multicolumn{2}{c}{Dawning} \\
      \includegraphics[scale=0.08]{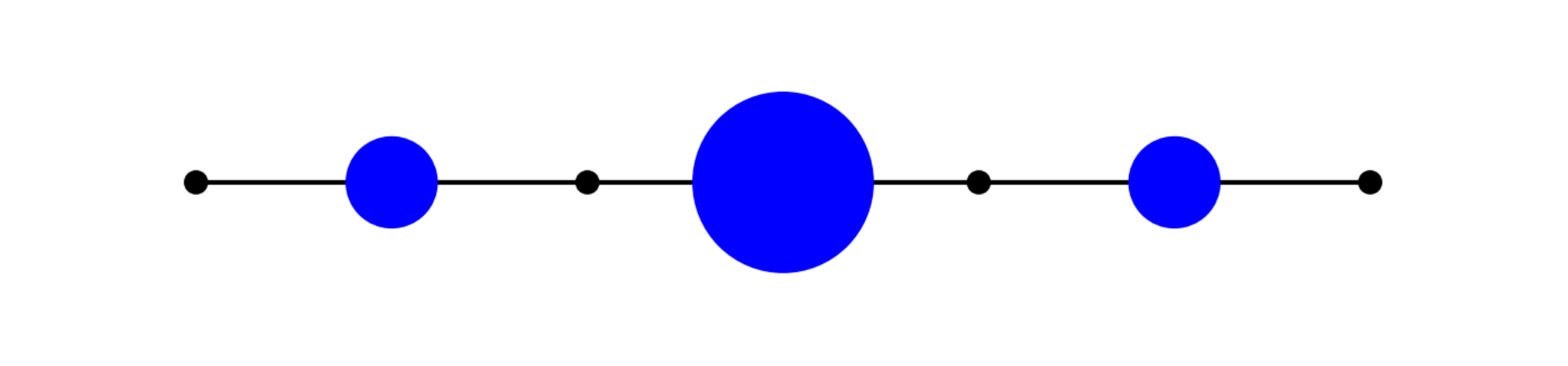}
       & \raisebox{3.1mm}{F} \\
       \multicolumn{2}{c}{Core--periphery} \\
       \noalign{\vskip 1.5ex}
      \includegraphics[scale=0.08]{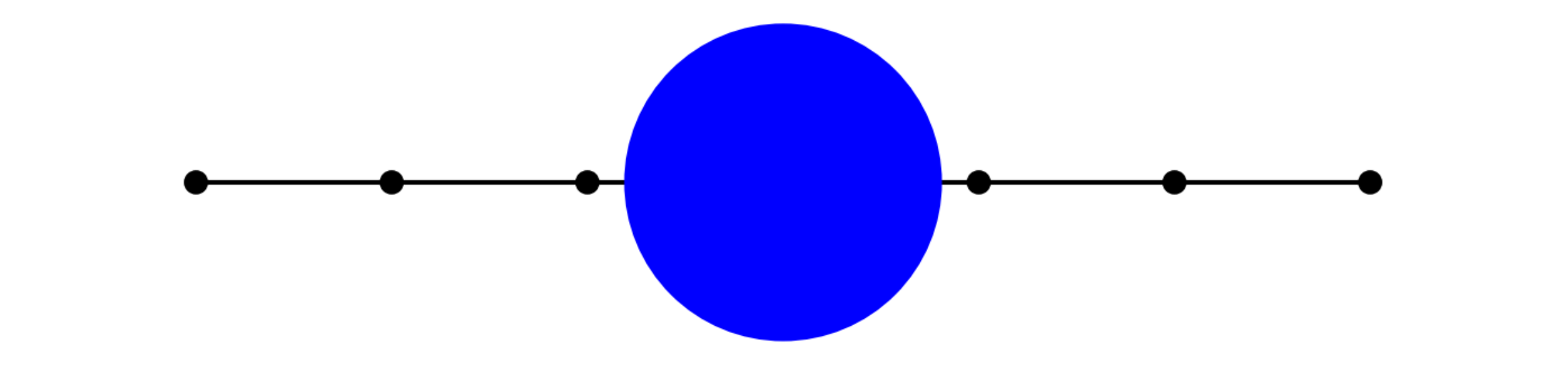}
       & \raisebox{3.1mm}{IJ,~K,~L} \\
       \multicolumn{2}{c}{Full agglomeration} \\
       \noalign{\vskip 1.5ex}
      \includegraphics[scale=0.08]{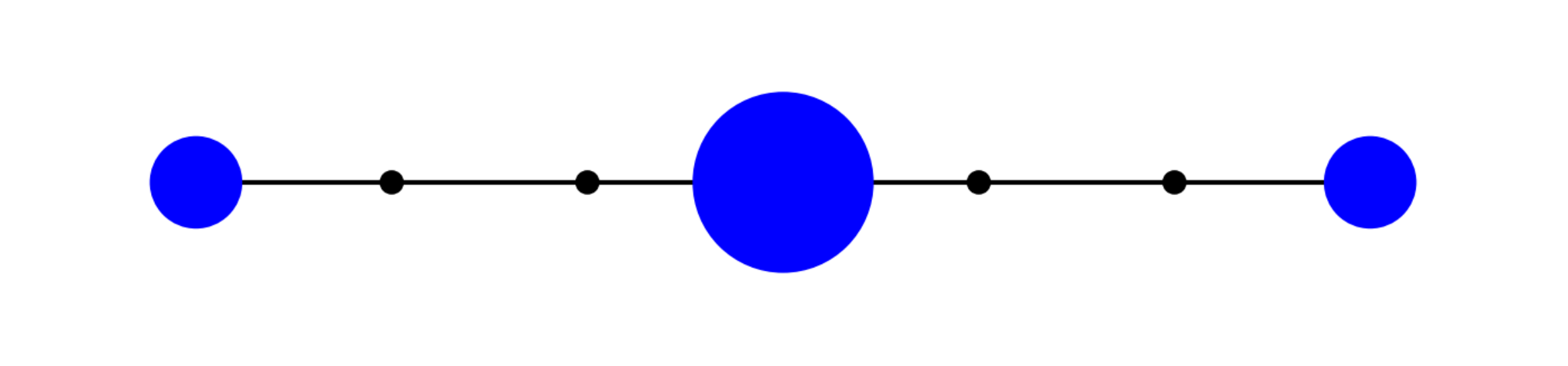}
       & \raisebox{3.1mm}{K'} \\
      \noalign{\vskip -2ex}
      \includegraphics[scale=0.08]{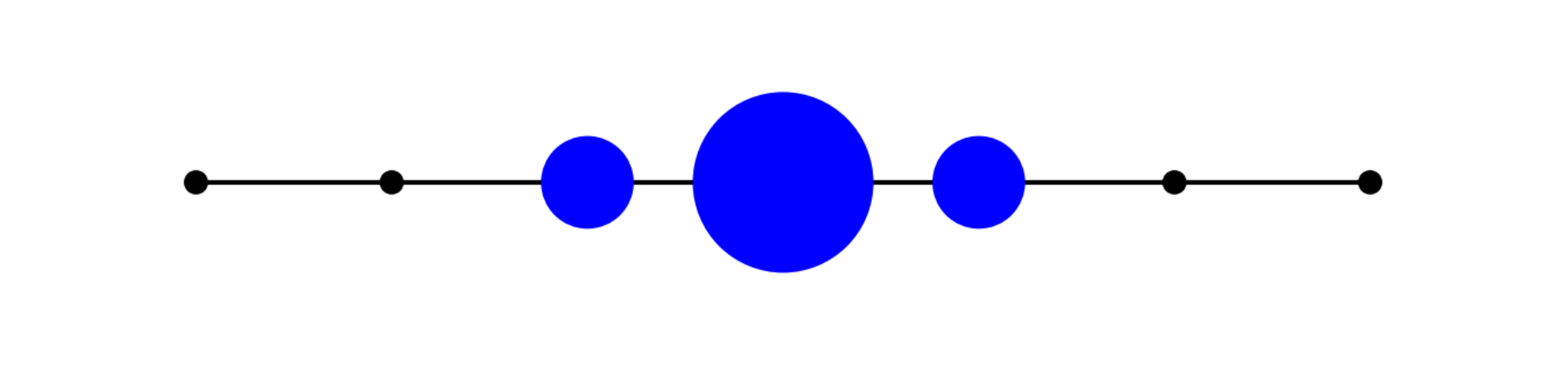}
       & \raisebox{3.1mm}{L'} \\
      \noalign{\vskip -2ex}
      \includegraphics[scale=0.08]{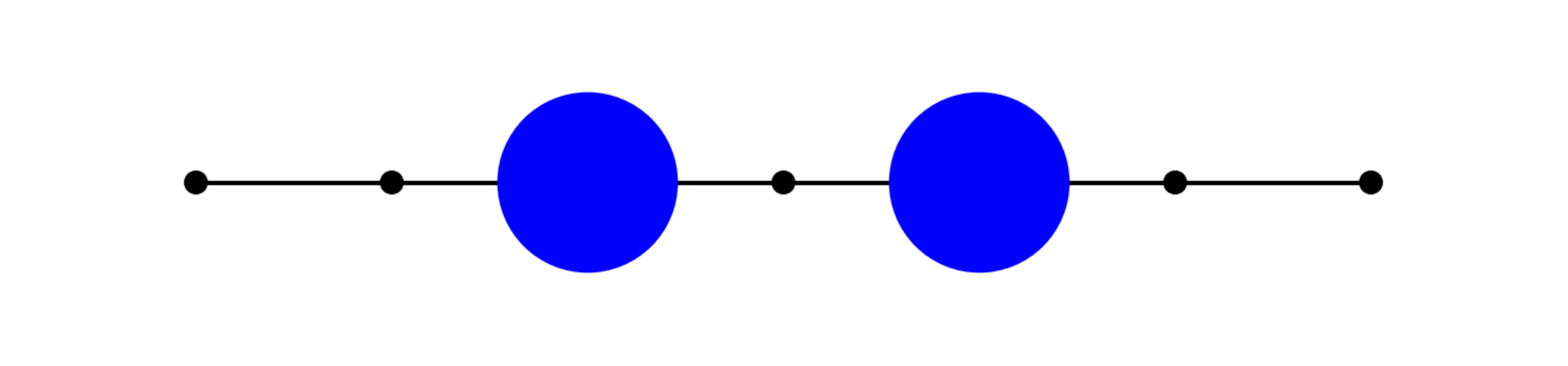}
       & \raisebox{3.1mm}{QR} \\
       \noalign{\vskip -1ex}
       \multicolumn{2}{c}{Other stages} \\
       \end{tabular}
  \end{center}
   \end{minipage}
   \end{small}
    	\caption{Solution curves of $K=7$ places
    	    	for the FO model with
  $(\sigma, \mu) = (6.0, 0.4)$
(solid line: stable; broken line: unstable; $\triangle$: bifurcation point;
$\circ$: sustain point)}
\label{K=7,sigma=6}
\vspace{-3mm}
\end{figure}
\FloatBarrier

\begin{figure}[!p]
\begin{small}
\begin{minipage}{.52\textwidth}
 \begin{center}
      \includegraphics[scale=0.16]{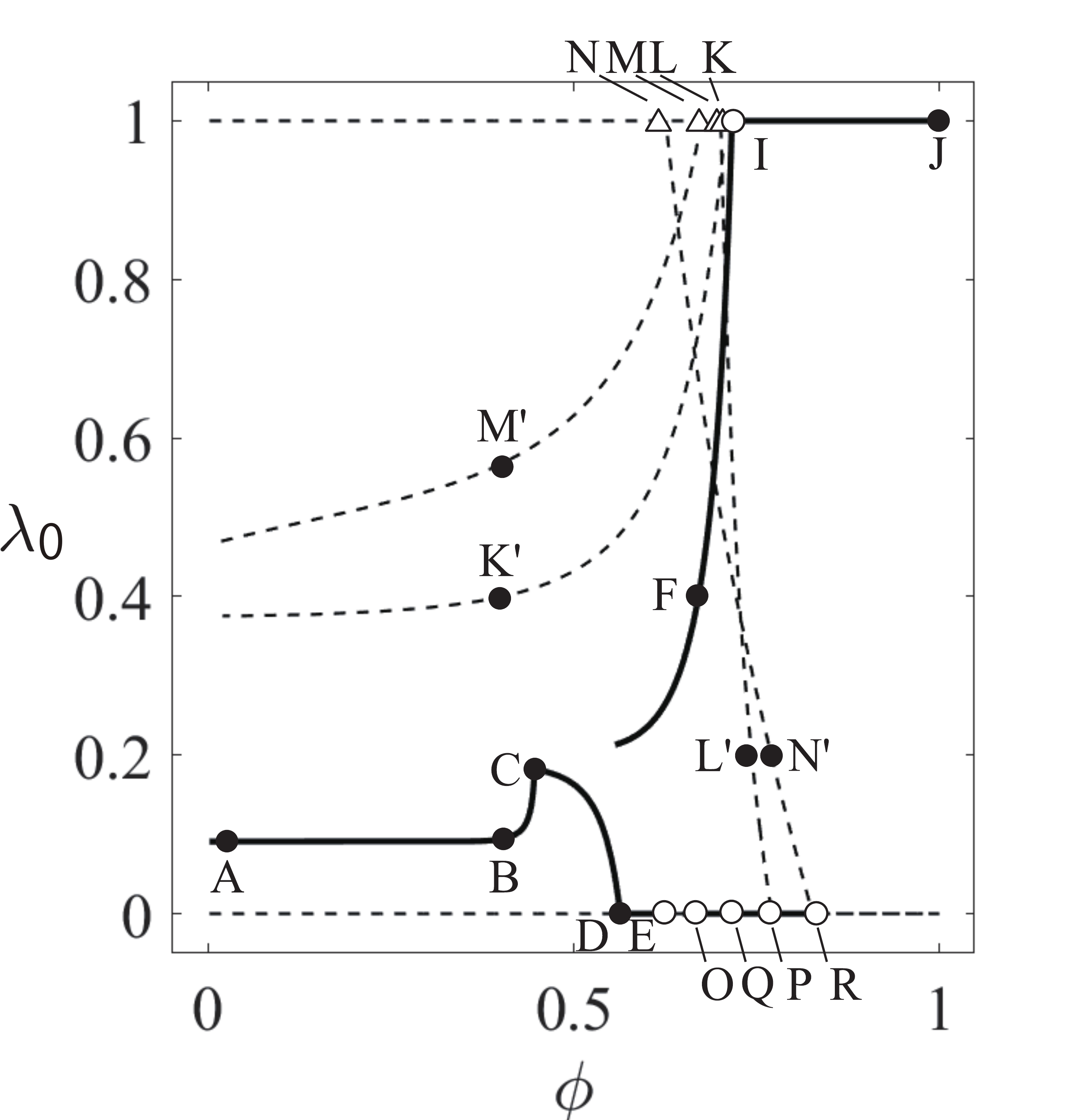} 
   \end{center}
   \end{minipage}
   \begin{minipage}{.45\textwidth}
  \begin{tabular}{l@{\hspace{-1mm}}l}
      \includegraphics[scale=0.08]{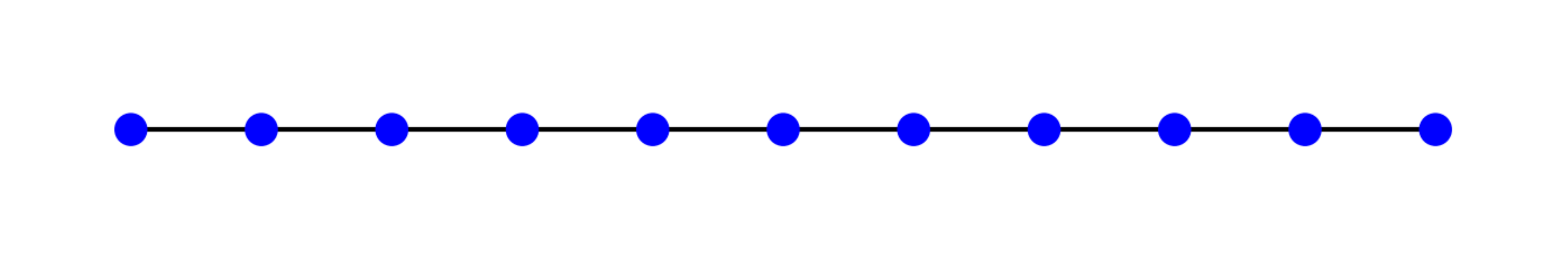}
       & \raisebox{3.5mm}{A} \\
      \includegraphics[scale=0.08]{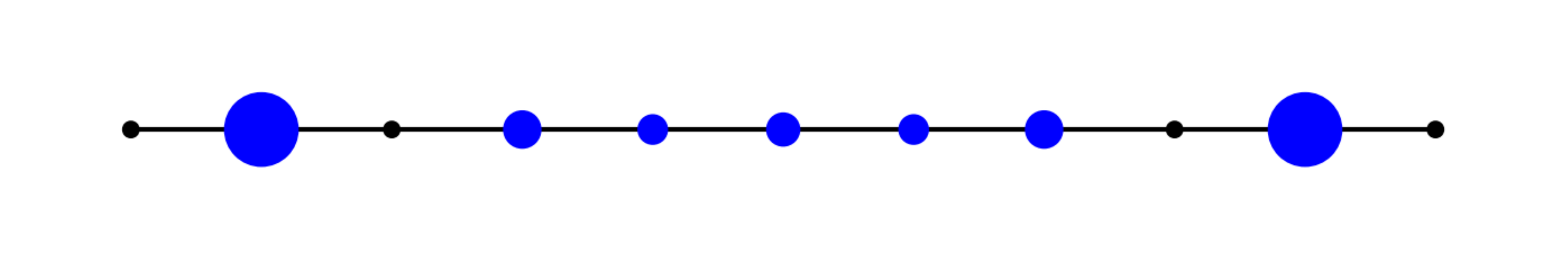} 
      & \raisebox{3.5mm}{B} \\
      \includegraphics[scale=0.08]{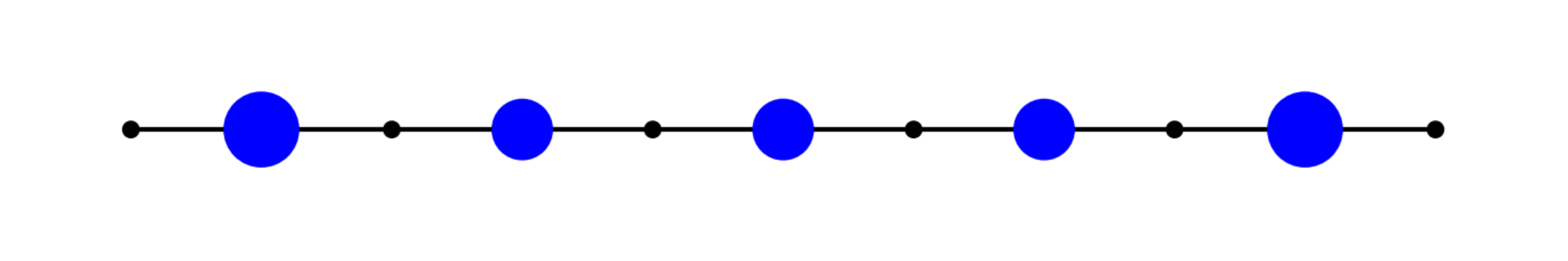}
      & \raisebox{3.5mm}{C} \\
      \includegraphics[scale=0.08]{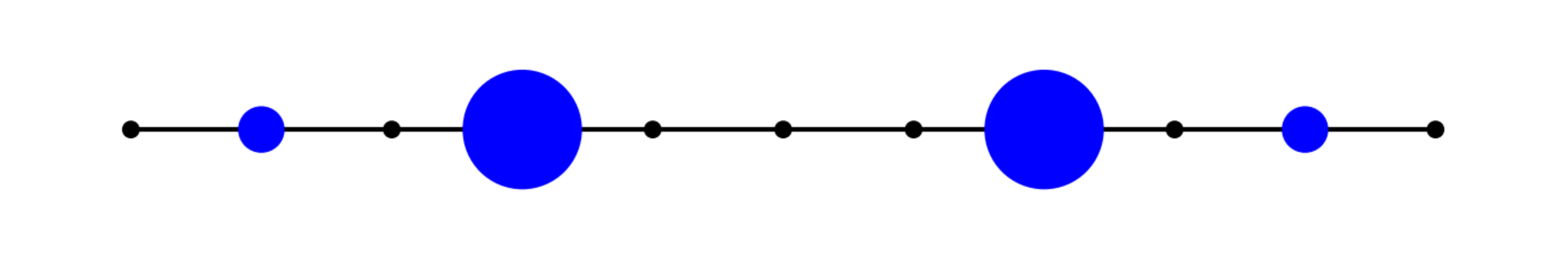}
       & \raisebox{3.5mm}{D} \\
      \includegraphics[scale=0.08]{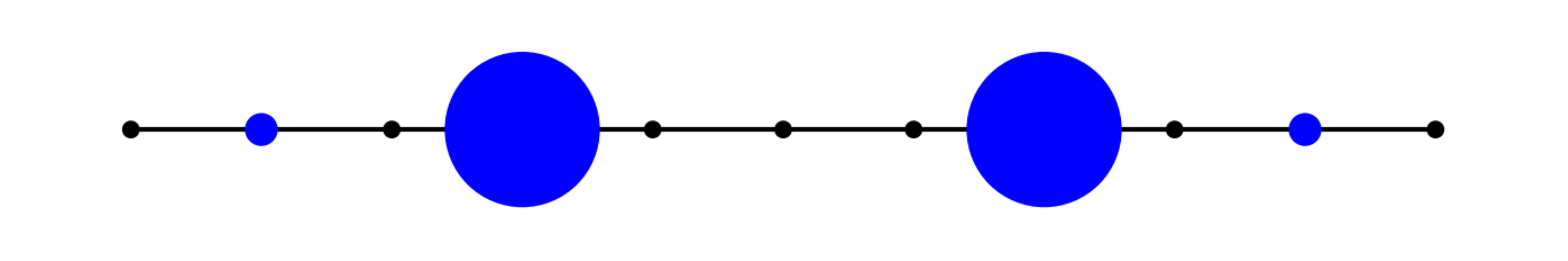}
       & \raisebox{3.5mm}{E} \\
    \multicolumn{2}{c}{Dawning} \\
     \includegraphics[scale=0.08]{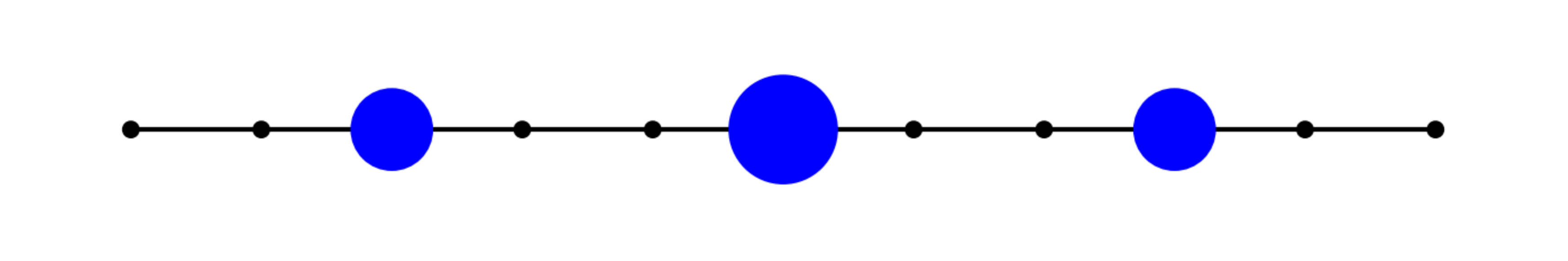}
       & \raisebox{3.5mm}{F} \\
    \multicolumn{2}{c}{Core--periphery} \\
       \noalign{\vskip 1ex}
    \includegraphics[scale=0.08]{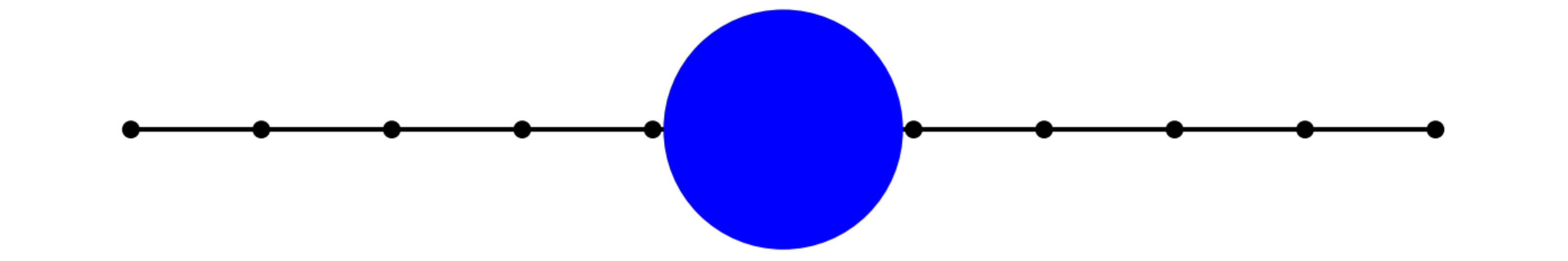}
       & \raisebox{3.5mm}{IJ} \\
      \multicolumn{2}{c}{Full agglomeration}
      \end{tabular} 
 \end{minipage}

  \begin{tabular}{l@{\hspace{-1mm}}l}
      \includegraphics[scale=0.08]{fig12_7_-eps-converted-to.pdf}
       & \raisebox{3.5mm}{K, L, M, N} \\
      \includegraphics[scale=0.08]{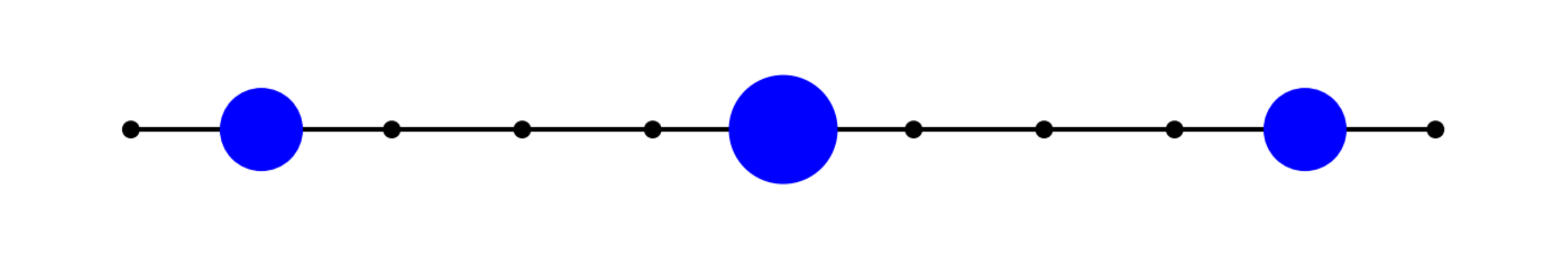} 
      & \raisebox{3.5mm}{K'} \\
      \includegraphics[scale=0.08]{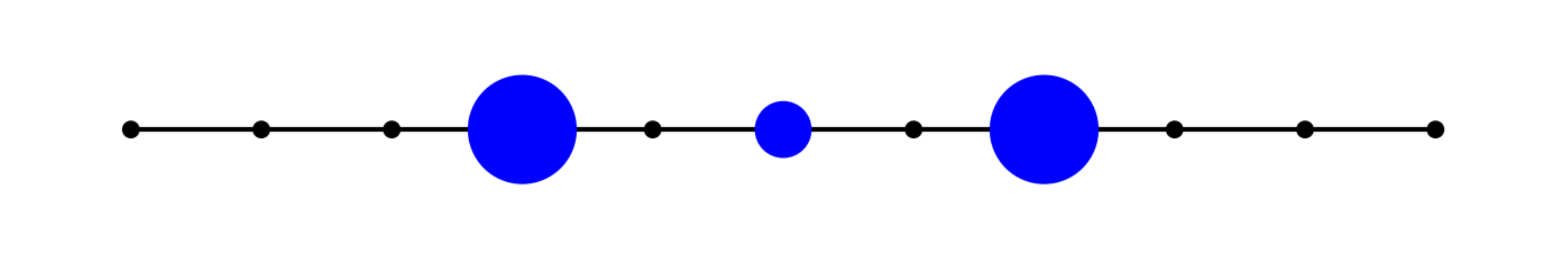}
      & \raisebox{3.5mm}{L'} \\
      \end{tabular} 
  \begin{tabular}{l@{\hspace{-1mm}}l}
      \includegraphics[scale=0.08]{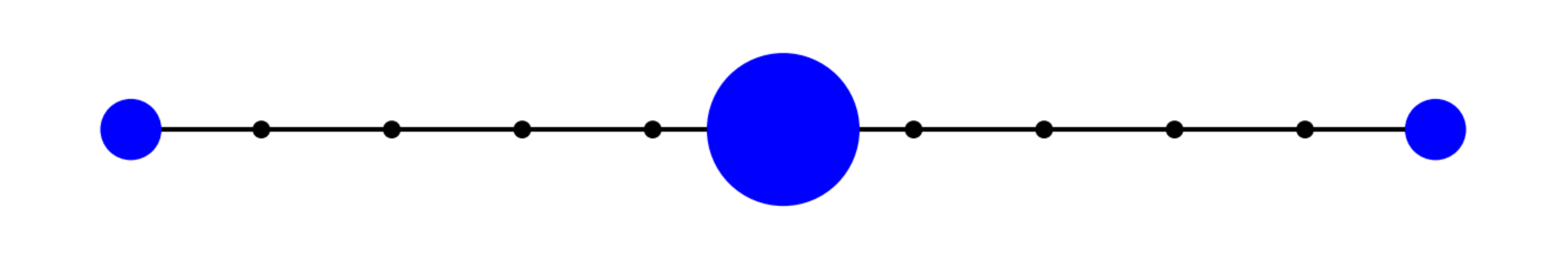}
      & \raisebox{3.5mm}{M'} \\
      \includegraphics[scale=0.08]{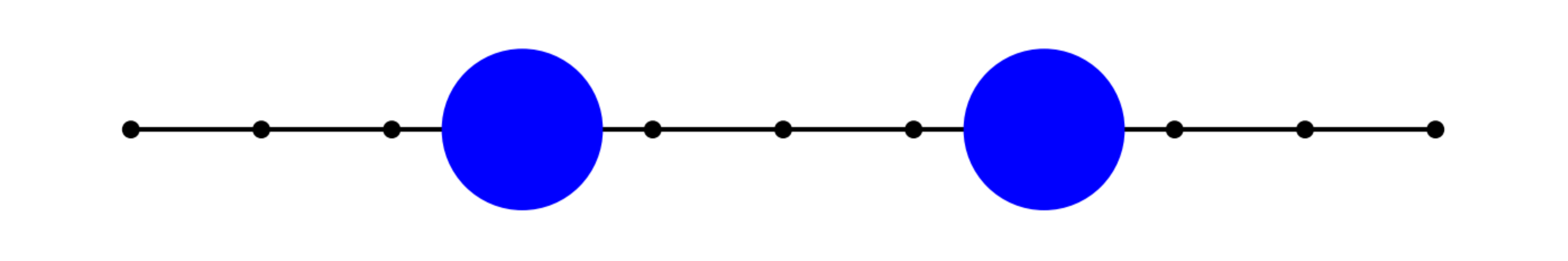}
      & \raisebox{3.5mm}{OP} \\
      \includegraphics[scale=0.08]{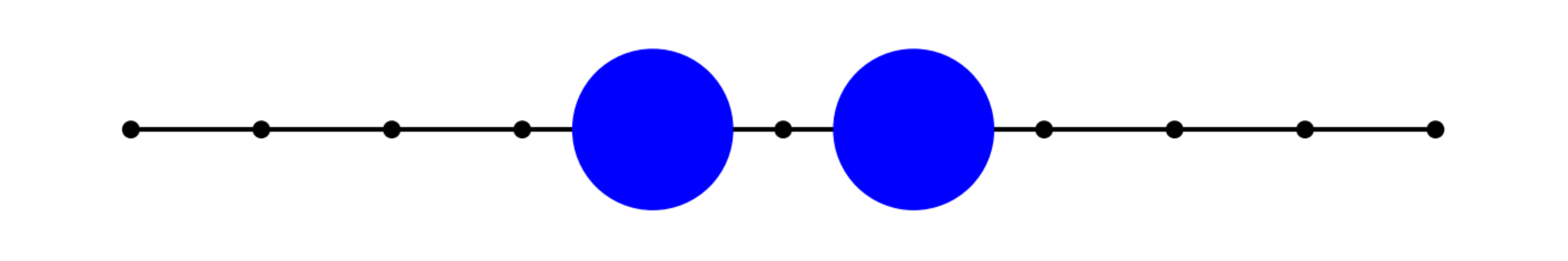}
      & \raisebox{3.5mm}{QR} \\
      \end{tabular} 
\end{small}
 	\caption{Solution curves of $K=11$ places
 	    	    	for the FO model with
 $(\sigma, \mu) = (6.0, 0.4)$
(solid line: stable; broken line: unstable; $\triangle$: bifurcation point;
$\circ$: sustain point)}
	\label{K=11,sigma=6}

\begin{small}
 \begin{minipage}{.58\textwidth}
  \begin{center}
      \includegraphics[scale=0.16]{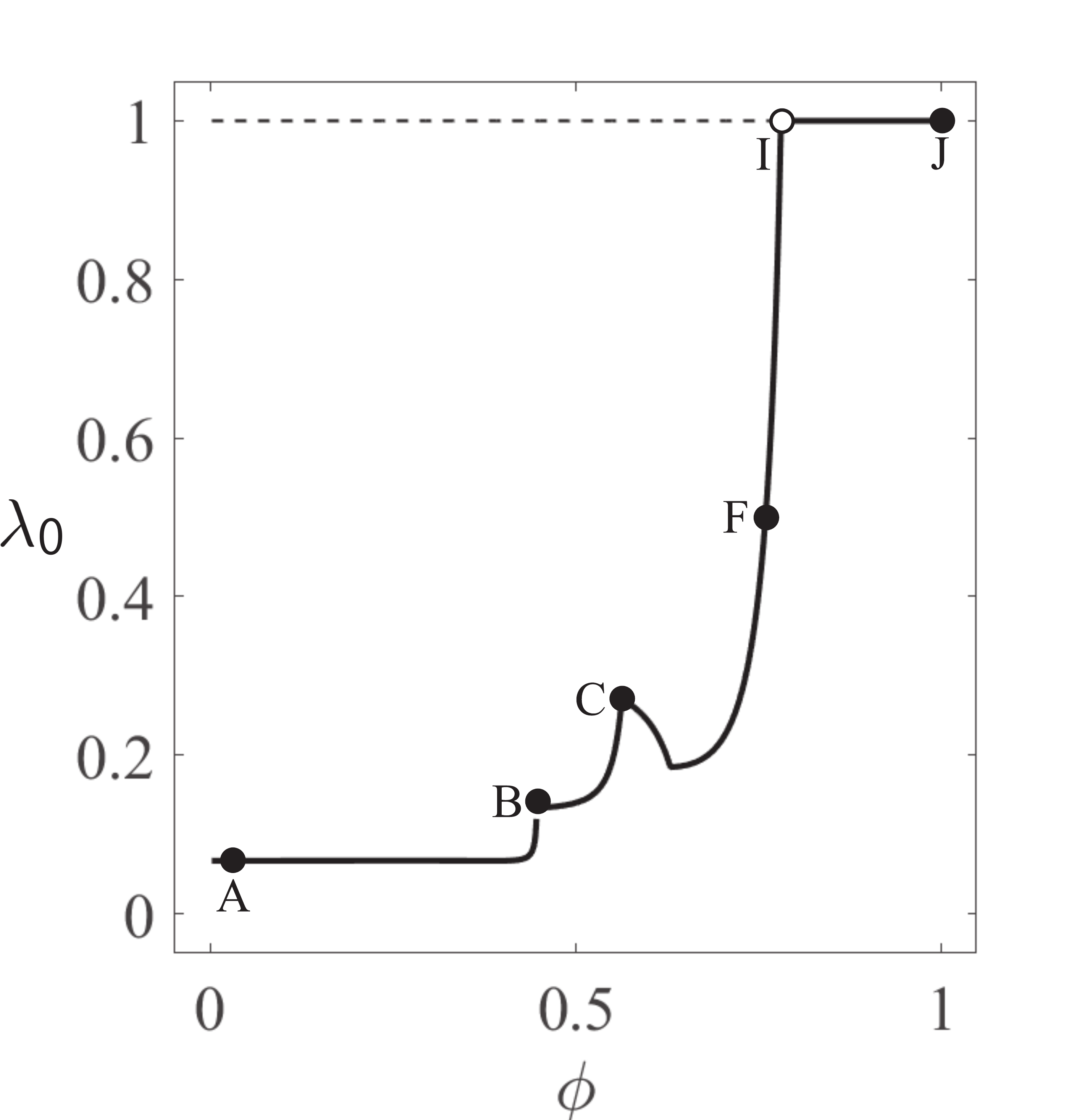} 
   \end{center}
   \end{minipage}
   \hspace{-3mm}
  \begin{minipage}{.32\textwidth}
    \begin{center}
  \begin{tabular}{c@{\hspace{-1mm}}l}
      \includegraphics[scale=0.06]{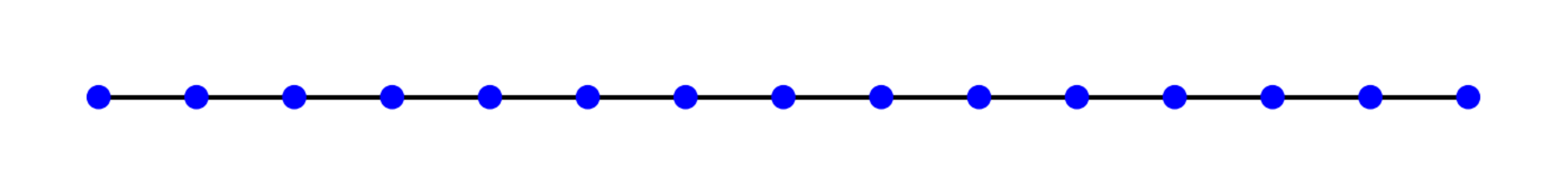}
       & \raisebox{1.75mm}{A} \\
      \includegraphics[scale=0.06]{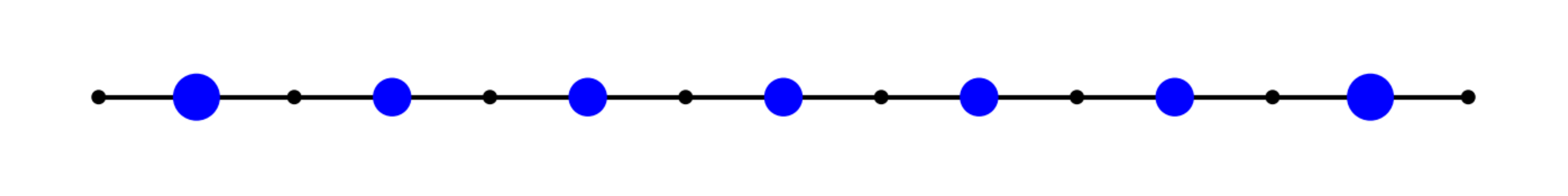} 
      & \raisebox{1.75mm}{B} \\
      \includegraphics[scale=0.06]{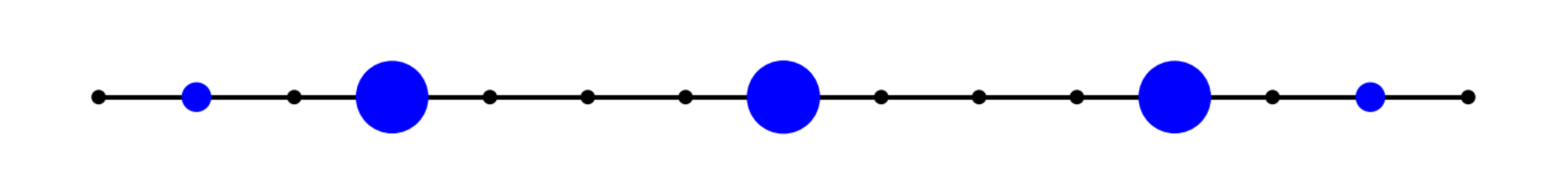}
      & \raisebox{1.75mm}{C} \\
     \multicolumn{2}{c}{Dawning} \\
       \noalign{\vskip 1ex}
      \includegraphics[scale=0.06]{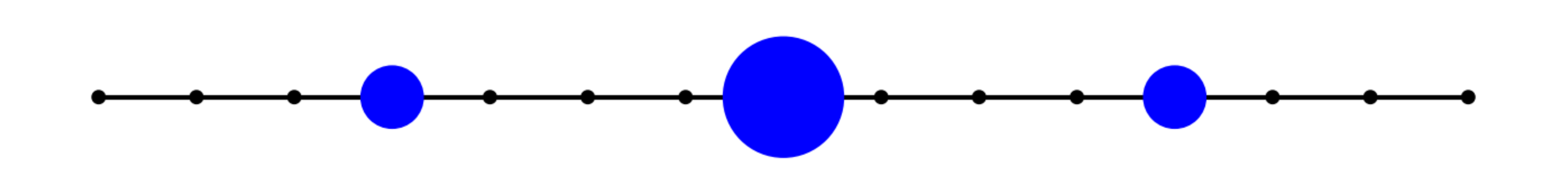}
       & \raisebox{1.75mm}{F} \\
       \multicolumn{2}{c}{Core--periphery} \\
       \noalign{\vskip 2ex}
      \includegraphics[scale=0.06]{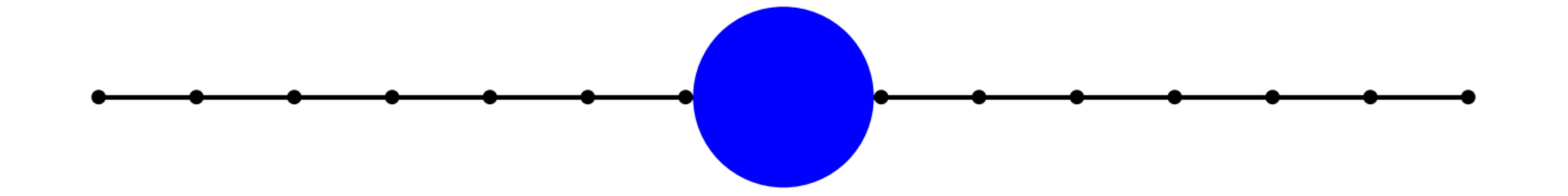}
       & \raisebox{1.75mm}{IJ} \\
       \multicolumn{2}{c}{Full agglomeration (stable)}
       \end{tabular}
  \end{center}
   \end{minipage}
\end{small}
    	\caption{Solution curves of $K=15$ places 
    	    	    	for the FO model with
  $(\sigma, \mu) = (6.0, 0.4)$
(solid line: stable; broken line: unstable; 
$\circ$: sustain point)}
\label{K=15,sigma=6}
\end{figure}
\FloatBarrier

The solution curves for $K=5$, 7, 11, and 15
are plotted in Figs.~\ref{K=5,sigma=6}--\ref{K=15,sigma=6}.
First, we investigate the bifurcating paths from the full agglomeration $\bm{\lambda}^{\rm FA}_0$.
The horizontal line at $\lambda_0=1$ is
the solution curve for this full agglomeration.
The stable full agglomeration (shown by the solid horizontal line IJ)
has the sustain point I (shown by $(\circ)$),
irrespective of the number $K$ of places.
The bifurcation behavior at this sustain point is
 in line with
  Proposition~\ref{Stability_Bifcurve}, 
  which states that there is at most one stable bifurcating path:
\begin{itemize}
\item
For $K=5$ places 
(Fig.~\ref{K=5,sigma=6}),
both bifurcating paths are unstable just after the bifurcation.
A bifurcating path IGF 
regains stability at the limit point of $\phi$ (point G shown as $(\square)$),
where twin satellite cities gain a significant population size.
There occurs a dynamic jump
between points I and G.
\item
The bifurcating path IF with twin satellite cities
is stable for $K=7,11,15$ places.
\end{itemize}
Thus, the bifurcating path with twin satellite cities
is superior in stability and is economically important.
Such superior stability might be due to balanced economic activities
on both sides of the economy. 
A stable core--periphery pattern emerges at point F.

Next, we turn our attention to the twin cities.
The solution curves for the twin cities are the horizontal line at $\lambda_0=0$.
For $K=5$
(Fig.~\ref{K=5,sigma=6}),
there is a stable path DE
for twin cities $\bm{\lambda}=\bm{\lambda}^{\rm Twin}_{1}
=(0,1/2,0,1/2,0)$,
enclosed by the two sustain points D and E.
The sustain point D 
corresponds to the case treated in Proposition \ref{twin sustain double};
there is a stable bifurcating path DC,
on which the population of the two border places ($i=\pm 2$)
 becomes positive. 
 Another sustain point E has the stable bifurcating path EFG,
on which the central place regains population
leading to an agglomeration into three cities
(Proposition \ref{twin sustain Bifurcation}).

Last, we observe changes in stable solutions
in Figs. \ref{K=5,sigma=6}--\ref{K=15,sigma=6}
as $\phi$ increases from 0 to 1.
There are three stages: the dawning stage, the stage with a stable core-periphery pattern, and the stage with a stable full agglomeration at the center.
The dawning stage is the stage where a nearly uniform population distribution, which is stable for small
$\phi$, changes to a stable core-periphery pattern.
In this stage, the nearly uniform distribution reaches this core-periphery pattern by continuously tracing stable equilibria or undergoing a dynamic jump from a stable equilibrium.

For $K = 7$ places in Fig.~\ref{K=7,sigma=6},
a nearly uniform population distribution prevails for small $\phi$ (point A), transitioning
to a distribution with significant populations at $i = 0$ and $i = \pm2$ in the dawning stage (points B and C).
The central city ($i = 0$) grows to engender a core--periphery pattern (point F) with a large central city surrounded by twin cities at $i = \pm2$. The growth of this central city produces the full agglomeration at the
center (path IJ). The changes in stable solutions for $K = 11$ and $15$ places
(Figs. \ref{K=11,sigma=6} and \ref{K=15,sigma=6}) also exhibit the dawning stage, the core--periphery pattern, and the full agglomeration in this sequence.
A state like twin cities appears in the dawning stage for $K = 11$ places (points D and E),
similarly to $K=5$ places.

We numerically obtain 
the distance $\delta_{\rm s}$ of the emerging twin cities from the center 
for various numbers $K$ of places.
As listed in Table \ref{list of core-sat},
the distance $\delta_{\rm s}$ enlarges as $K$ increases.
As an index for the optimal location of the satellite cities,
we consider a normalized distance 
 $\delta_{\rm s}/k$ with $k=(K-1)/2$.
As $k$ increases to a large value, such as $k=49$,
$\delta_{\rm s}/k$ converges to
a little beyond halfway from the center to the border
$(\delta_{\rm s}/k=28/49\approx 0.57)$.

For various values of parameter $\sigma\in\{2.5,6.0,10.0\}$
with the same value of $\mu=0.4$,
$\delta_{\rm s}/k \in\{0.70,0.57,0.49\}$ 
for $k=49$.
Smaller $\sigma$ enlarges $\delta_{\rm s}/k$
and the periphery places tend to be located 
further away from the center.
It seems that there might be a 
golden rule for $\delta_{\rm s}$
 in terms of $\mu$ and $\sigma$.

\begin{table}[htbp]
\centering
\caption{The values of $\delta_{\rm s}/k$ for various values of $k$.}
\label{list of core-sat}
\begin{tabular}{c|ccccccc} \hline
$k$ & 2 & 3 & 4 & 5 & 7 & $\cdots$ & 49 \\ \hline
$\delta_{\rm s}$ &
 1 & 2 & 3 & 3  & 4  & $\cdots$ & 28 \\
$\delta_{\rm s}/k$ & 0.5 & 0.667 & 0.75 & 0.6
 & $4/7\approx 0.57$  & $\cdots$ & $28/49\approx 0.57$
 \\ \hline
\end{tabular}
\end{table}


\section{Historical change in population in Japan's Main Island}\label{Use of EG analysis}

We have observed several theoretical patterns of twin peripheral cities
around the central city and ensured their existence
for the typical economic geography models. 
Taking Japan's Main Island (Fig.~\ref{USA=JAPAN}(a)) as a reference 
with $K=5$ places,
we demonstrate the usefulness of
the proposed procedures in the analysis of real population data.

In the analysis for $K=5$ shown in Fig.~\ref{K=5,sigma=6}, 
we can observe several kinds of population distributions 
at points A, B, $\ldots$, J.
The bifurcation occurs at the sustain point I, engendering 
twin cities at $i=\pm 1$
leading to a stable state A--G.
  As the trade freeness $\phi$ increases from 0 to 1,
there emerge three stable stages:
(1) the twin cities at $i=\pm1$ (Tokyo and Osaka)
 along the path DE for $\phi$ slightly below 0.5,
 (2) the central city surrounded by the twin cities
 in an intermediate stage along the path EG (containing   
 E$^\prime$, E$^{\prime\prime}$ and F), and
(3) the full agglomeration at the center (Nagoya)
along the horizontal path IJ for $\phi$ over 0.5.
One could infer that 
the configuration of Japan resembles that of the pattern E$^{\prime\prime}$:
Nagoya, a city at the center, has a smaller population compared to the 
gigantic twin cities of Osaka and Tokyo at $i=\pm 1$, 
but is much larger than the border cities
 of Hiroshima and Sendai at $i=\pm 2$.

When we compare the five cities' populations in the 1950s (shown by blue arcs in 
Fig.~\ref{USA=JAPAN}(a))
with the population in 2020 (shown by orange arcs),
we observe that Nagoya has grown 
significantly more compared to Tokyo and Osaka, thus suggesting that Japan may be in the stage 
of E$^\prime$E$^{\prime\prime}$, expressing the growth of the 
central city (Nagoya).
Such growth is in line with the recent numerical experiment on  
Shinkansen extension 
\cite{Hayakawa.2021},
which reports that the investment in the Shinkansen network
would lead to population change in large metropolitan areas of
Tokyo, Osaka, and Nagoya by
$-0.3\%$,
$0.6\%$, and $9.8\%$,
respectively.
Thus, 
the most significant growth of the population in response to the investment
is expected to occur in Nagoya, located
at the geographical center of the metropolitan areas of Tokyo and Osaka.


\section{Conclusion}\label{section conclusion}

We have elucidated the bifurcation mechanism of the emergence of 
twin cities around the central city in a long narrow economy.
This mechanism is independent of models 
and applicable to general spatial economic models
under replicator dynamics.
A pertinent combination of this model-independent general bifurcation 
mechanism with model-dependent properties, such as stability/sustainability
and parameter dependency,
would be vital in the 
understanding of the agglomeration behaviors
in the long narrow economy.

  We have observed diverse agglomeration patterns
dependent on the number of places,
trade freeness, and economic parameters.
Such dependence is a possible source of the diversity of
the population distribution of a chain of cities observed worldwide.
Using the analysis result for five cities, we explained the population distribution
in a chain of cities on Japan's Main Island.
Thus,
the analysis by the economic geography model
gives insight into the investigation of actual data.
  It motivated us to conduct  
  the study of the bifurcation mechanism
  from the full agglomeration presented in this paper.

A remark is on the standpoint of this paper. 
While it is customary to start from the uniform state,
we highlight agglomeration patterns emanating from the
fully agglomerated state.
Nowadays, it would be far more essential to investigate the competition
between a large central city and peripheral cities than
to investigate the self-organization of towns in a flat land
envisaged in central place theory.

Future research will elucidate how and where population agglomeration emerges in a two-dimensional space.
For example, for Southern Germany, which is
the motherland of central place theory
\cite{Christaller.1933},
a hexagonal lattice network would be a pertinent
spatial platform.
Indeed, hexagonal population distributions in actual data
of Germany and the Eastern United States are observed
using the hexagonal lattice network
\cite{Ikeda.etal.2018,Ikeda.NETS.2022}.
Another future work is to consider other economic indicators,
while this paper focuses on the size of the population.


\section*{Acknowledgments}

We gratefully acknowledge 
Grants-in-Aid received from 
Japan Society for the Promotion of Science
Grant/Award Number: 
21K04299;
and  
Funda\c{c}\~{a}o para a Ci\^{e}ncia e Tecnologia
        UIDB/04105/2020, UIDB/00731/2020, PTDC/EGE-ECO/30080/2017. Part of this research was developed while Jos\'{e} M. Gaspar was a researcher at the Research Centre in Management and Economics, Católica Porto Business School, Universidade Católica Portuguesa, through the grant CEECIND/02741/2017.


\appendix


\setcounter{figure}{0}
\setcounter{table}{0}
\setcounter{equation}{0}
 \renewcommand{\thefigure}{A\arabic{figure}}
 \renewcommand{\thesection}{A}
 \renewcommand{\thetable}{A\arabic{table}}
\renewcommand{\theequation}{A\arabic{equation}}

\section{Theoretical details of bifurcation analysis}\label{Proof_Section 3}

\subsection{Proof of Proposition~\ref{Invariant proposition}}\label{Proof FA twin}

For $\bm{\lambda}=\bm{\lambda}^{\rm FA}_\delta$,
we have $\lambda_i=0$ $(i\neq -\delta)$ 
and $v_{-\delta}-\bar{v}=0$ since $\bar{v}=v_{-\delta}$;
accordingly, the governing equation \eqref{gov.eq.general.F}
with \eqref{F_expression} is satisfied for any $i \in N$. 
For $\bm{\lambda}^{\rm Twin}_\delta$,
we have $v_{-\delta}=v_{\delta}$ by symmetry, 
$v_{-\delta}-\bar{v}=v_{-\delta}-\frac12 v_{-\delta}-\frac12 v_{\delta}=0$
and similarly $v_{\delta}-\bar{v}=0$.
We also have 
$\lambda_i=0$ $(i\neq \pm \delta)$;
accordingly, 
the governing equation is satisfied.

\subsection{Proof of Proposition~\ref{BifurCurveV_2}}\label{Proof Solutions}

For $(\lambda_{-j},\lambda_j)=(0,0)$, 
the reduced governing equation \eqref{Eqs pm j} is satisfied for any $\phi$.

For $(\lambda_{-j},\lambda_j)=(x,x)$ for $x>0$,
\eqref{symmetry conditions App}
gives 
$v_{-j}(x,x,\phi)=v_j (x,x,\phi)$.
Then, the two equations in 
 \eqref{Eqs pm j} become identical
and are satisfied by a solution curve
$(\lambda_{-j},\lambda_j,\phi)=(x,x,\phi)$
satisfying $v_j (x,x,\phi)-\bar{v} (x,x,\phi)=0$.
Appendix \ref{Stability Local Sustain FA}
shows the existence of this curve in the neighborhood of the sustain point.

For $(\lambda_{-j},\lambda_j)=(x,0)$ for $x>0$,
the second of the two equations in 
 \eqref{Eqs pm j} is always satisfied
and the first one is satisfied by
a solution curve
$(\lambda_{-j},\lambda_j,\phi)=(x,0,\phi)$
satisfying $v_{-j} (x,0,\phi)-\bar{v} (x,0,\phi)=0$.
Another case of
$(\lambda_{-j},\lambda_j)=(0,x)$ for $x>0$
can be treated similarly.

\subsection{Proof of 
Propositions~\ref{Stability_Bifcurve}
and \ref{Stability_of_Bifcurve}}\label{Stability Local Sustain FA}

To investigate the stability of the bifurcating solutions at the sustain point
$(\bm{\lambda}^{\rm FA}_0, {\phi^{\rm c}})$,
we introduce the incremental form of the governing equation
\eqref{Eqs pm j}.
We consider 
two independent variables $(x,y)=(\lambda_{-j},\lambda_{j})$
with an incremental parameter $\psi  = \phi - \phi^{\rm c}$.
Then, in the neighborhood of the sustain point, the governing equation \eqref{Eqs pm j}
with the second symmetry condition
in \eqref{symmetry conditions App} has an expanded form.
  \begin{align}
  \begin{array}{l}
    x \, (a \psi + b x + c y+\rm{higher~order~terms})=0,
    \\
    y \, (a \psi + b y + c x+\rm{higher~order~terms})=0
    \end{array}
        \label{BifEqu}
  \end{align}
with some expansion coefficients $a,b,c$.  

It is apparent from this equation that there are bifurcating solution curves.
For $(x,y)=(w,w)$, we have the bifurcating curve
$\psi=-(b+c)w/a + \rm{higher~order~terms}$,
while $\psi = \psi_2 \approx -b w/a + \rm{higher~order~terms}$ for $(x,y) = w(1,0)$.

The Jacobian matrix for this reduced governing equation \eqref{BifEqu} reads
\begin{align*}
 \hat{J} \approx
 \left(
  \begin{array}{cc}
    a \psi + 2b x + c y & c x
    \\  
    c y & a \psi + 2b y + c x
  \end{array}  \right).
\end{align*}

\noindent
The use of $(x, y) =(w,w)$
and $\psi=\psi_1 \approx -(b + c)w/a$ 
in $\hat{J}$ leads to $\hat{J}_1$ below
and the use of $(x,y) = (w,0)$ and $\psi = \psi_2 \approx -b w/a$
leads to $\hat{J}_2$ as
\[ \hat{J}_1 \approx w \left(
  \begin{array}{cc}
    b & c
    \\
    c & b
  \end{array}
  \right),
  \qquad
   \hat{J}_2 \approx w \left(
  \begin{array}{cc}
    b & c
    \\
    0 & c - b
  \end{array}
  \right).
  \]
\begin{lemma}\label{Eigen_Bifcurve}
The bifurcating solution $(\Delta\bm{\lambda}_1, \psi_1)$ has the eigenvalues:
    $e_1 \approx (b + c) w$ and
    $e_2 \approx (b - c) w$.
Another solution
$(\Delta\bm{\lambda}_2, \psi_2)$ has the eigenvalues:
    $e_1 \approx b w$ and
    $e_2 \approx (c - b) w$.
\end{lemma}

\begin{lemma}\label{Stability_Bifcurve Lemma}
Under the assumption that the state of the full agglomeration
is stable for $\psi<0$ (respectively, $\psi>0$), 
there are three cases: 
(i) If $ -b > |c|,$ only the first bifurcating path $(\Delta\bm{\lambda}_1, \psi_1)$ is stable.
(ii) If $c < b < 0$, only the second bifurcating path 
     $(\Delta\bm{\lambda}_2, \psi_2)$ is stable.
(iii) Otherwise, both paths are unstable.
(iv) A stable bifurcating path branches in the direction of 
$\psi<0$ (respectively, $\psi>0$).
\end{lemma}
\begin{proof}
For the fully agglomerated state $(x, y) = (0,0)$, 
we have $\hat{J} = a \psi I$ with the eigenvalue $a \psi$ (twice repeated).
If the state is sustainable for $\psi>0$ (respectively, $\psi<0$), 
we have $a < 0$ (respectively, $a>0$).
(i)
The first bifurcating solution $(\Delta\bm{\lambda}_1, \psi_1)$
with 
  $e_1 \approx (b + c) w$ and
  $e_2 \approx (b - c) w$
(cf., Lemma~\ref{Eigen_Bifcurve})
is stable if 
$-b > |c|$.
Since $b + c < 0$, $a < 0$, and $w> 0$,
 $\psi=\psi_1 \approx -(b + c)w/a$ 
gives $\psi=\psi_1 < 0$
(respectively, $\psi_1>0$).
(ii)
The second bifurcating solution $(\Delta\bm{\lambda}_2, \psi_2)$
with $e_1 \approx b w$ and $e_2 \approx (c - b) w$ $(w>0)$
is stable if
$c < b < 0$.
Since $b < 0$, $a < 0$ and $w> 0$,
$\psi=\psi_2 \approx -b w/a$ 
gives $\psi=\psi_2 < 0$ (respectively, $\psi_2>0$).
The two bifurcating solutions cannot be stable simultaneously
since $-b > |c|$ and $c < b < 0$ are contradictory.
(iii) and (iv) are apparent from (i) and (ii).
 \end{proof}


\setcounter{figure}{0}
\setcounter{table}{0}
\setcounter{equation}{0}
 \renewcommand{\thefigure}{B\arabic{figure}}
 \renewcommand{\thesection}{B}
 \renewcommand{\thetable}{B\arabic{table}}
\renewcommand{\theequation}{B\arabic{equation}}

\section{Theoretical details 
of the full 
agglomeration at the center}\label{Indirect utilities and sustain}

\subsection{Common features for the FO and the PFSU 
models}\label{FE and PFSU App}

We can obtain $w_i$ and $\Delta_i$, which are represented in indirect utilities for the FO model and PFSU model (i.e., (\ref{IndiUtiFE}) and (\ref{IndiUtiPFSU}), respectively) as follows.
For the full agglomeration $\bm{\lambda}^{\rm FA}_0$
($\lambda_{0}=1$ and $\lambda_j=0$ for $j\neq 0$),
we have 
\begin{align}\label{AppDeltaFullAgg}
	\Delta_i (\bm{\lambda}) = 
	\left\{
	\begin{array}{ll}
		1 & \quad (i = 0),
		\\
		\phi^{|i|} & \quad (i \neq 0).
	\end{array}
	\right.
\end{align}
Furthermore, the wage in place $i\in N$ is given by a general form:
\begin{equation}
	w_{i}(\boldsymbol{\lambda})=
	\frac{\xi}{\sigma}
	\sum_{p\in N}
	\dfrac{\phi_{ip}b_{p}}{\sum_{m\in N} \phi_{mp}\lambda_{m}},
	\label{eq:generalwageequation}
\end{equation}
where $\xi$ and $b_{j}$ are model dependent
and
$\xi=\mu$ and
$b_{i}=\left(1+w_{i}\lambda_{i}\right)$
for the FO model and 
$\xi=\alpha$ and $b_{i}=\left(1+\lambda_{i}\right)$
for the PFSU model.
The wage in place
$j \neq 0$ 
$(b_j=1)$
is given by
\begin{equation}
	w_{\pm j}
	=\frac{\xi}{\sigma}\left(\phi^{j}b_{0}+\sum_{p \in N \backslash \{ 0\} } \dfrac{\phi_{jp}}{\phi_{0p}}\right)
	\quad	(j=1,\ldots,k).
	\label{eq:wagegeneralform}
\end{equation}

\noindent
The second term of the nominal wage is systematically given,
using $\phi_{jp}=\phi^{|j-p|}$, by
\begin{align*}
	\sum_{p \in N \backslash \{ 0\}}\dfrac{\phi_{jp}}{\phi_{0p}} & 
	=
	\sum_{p=-k}^{-1}\frac{\phi_{jp}}{\phi_{0p}} 
	+ \sum_{p=1}^{j}\frac{\phi_{jp}}{\phi_{0p}}
	+ \sum_{p=j+1}^{k}\frac{\phi_{jp}}{\phi_{0p}}
 =
 \sum_{p=-k}^{-1}\frac{\phi^{j-p}}{\phi^{-p}}
 + \sum_{p=1}^{j}\frac{\phi^{j-p}}{\phi^{p}}
 + \sum_{p=j+1}^{k}\frac{\phi^{p-j}}{\phi^{p}}
 \\
	& =
	\sum_{p=-k}^{-1}\phi^{j}
	+ \sum_{p=1}^{j}\phi^{j-2p}
	+ \sum_{p=j+1}^{k}\phi^{-j}
	=
	k\phi^{j} 
	 + \frac{\phi^{j}-\phi^{-j}}{\phi^{2}-1}
	 + (k-j)\phi^{-j}.
\end{align*}

\noindent
Thus, the wage in place $j \neq 0$ is finally given by
\begin{equation}
	w_{\pm j}=\dfrac{\xi}{\sigma}
	\left[ \phi^{j}(k+b_{0})
	+\phi^{-j}(k-j)+\frac{\phi^{j}-\phi^{-j}}{\phi^{2}-1}\right]
\quad	(j=1,\ldots,k).
	 \label{eq:nominal wage final}
\end{equation}

\subsection{The FO model}\label{The FE model Appendix}

{\bf Derivation of $v_i$
(Proof of \eqref{Payoff_full_Agglomeration_Center}
and \eqref{Payoff_full_Agglomeration}):}
We have $b_{i}=\left(1+w_{i}\lambda_{i}\right)$.
By setting $i=0$ in \eqref{eq:generalwageequation},
we obtain the wage in the central city:
\[
w_{0}=\frac{\hat{\mu}}{1-\hat{\mu}}\left(2k+1\right),
\quad
\hat{\mu}=\frac{\mu}{\sigma}. 
\]

\noindent
The nominal wage in any potential city at
$j \ (\neq 0)$ is given by
\[
w_{\pm j}=\hat{\mu}\left\{ \phi^{j}\left[k+1+\frac{\hat{\mu}}{1-\hat{\mu}}\left(2k+1\right)\right]+\left[\phi^{-j}(k-j)+\frac{\phi^{j}-\phi^{-j}}{\phi^{2}-1}\right]\right\} \quad	(j=1,\ldots,k).
\]

\noindent
Substituting (\ref{AppDeltaFullAgg}) and the above wage equations into 
(\ref{IndiUtiFE}) yields the indirect utilities for 
the FO model in places 
$i=0$ and $\pm j \ (1\le j \le k)$:
	\begin{align*}
		& v_{0} =\log\frac{\hat{\mu}}{1-\hat{\mu}}+\log(2k+1),
		\nonumber \\
		 & v_{\pm j} =\log\frac{\hat{\mu}}{1-\hat{\mu}}
		+\frac{j\mu\log\phi}{\sigma-1}+\log\left\{ \phi^{j}(\hat{\mu} k+k+1)+(1-\hat{\mu})\left[
		\phi^{-j}(k-j)
		+\frac{\phi^{j}-\phi^{-j}}{\phi^{2}-1}
\right]\right\}.
	\end{align*}

\noindent
These equations give \eqref{Payoff_full_Agglomeration_Center}
and \eqref{Payoff_full_Agglomeration}.

{\bf Existence of local sustain point
(Proof of Lemma \ref{Sustainable_Full_Agglomeration_NEW}(i)):}
In the proof of the existence (Lemma~\ref{Existence-Lemma}) and the uniqueness
(Lemma~\ref{Existence-Lemma-F}) of local sustain point
	$\phi_{j}^{c}\in(0,1)$,
we introduce several 
positive constants:
\begin{align*}
	a=\sigma-\mu>0, \quad
	b=\sigma-\mu-1>0, \quad
	c=\sigma+\mu-1>0, \quad
	d=\sigma-1>0.
\end{align*}
The sustain point must satisfy
$\mathcal{S}^{{\rm FO}}\equiv v_{\pm j}-v_{0}=0$
for some $j~(\ge 1)$
with 
	\begin{equation}
		\mathcal{S}^{{\rm FO}}=
		\frac{j\mu\log\phi}{\sigma-1}+\log\left\{ \phi^{j}(\hat{\mu} k+k+1)+(1-\hat{\mu})\left[
		\phi^{-j}(k-j)
		+\frac{\phi^{j}-\phi^{-j}}{\phi^{2}-1}
		\right]\right\} -\log(2k+1).\label{eq:utility difference}
	\end{equation}
\begin{lemma}\label{Existence-Lemma}
There exists at least
	one local sustain point
	$\phi_{j}^{c}\in(0,1)$ such that $\mathcal{S}^{{\rm FO}}(\phi_{j}^{c})=0$.
\end{lemma}
\begin{proof}
Notice that
\begin{align}
 & \lim_{\phi\rightarrow +0}\mathcal{S}^{{\rm FO}}=+\infty,
 \quad 
 \lim_{\phi\rightarrow 1}\mathcal{S}^{{\rm FO}}=0.
\label{SFElim} \end{align}
	Differentiating $\mathcal{S}^{{\rm FO}}$ with respect to $\phi$, we get
	\[
	\frac{\partial\mathcal{S}^{{\rm FO}}}{\partial\phi}=\frac{A_{1}-A_{2}}{d\phi\left(\phi^{2}-1\right)A_{3}},
	\]
	where
	\begin{align*}
		A_{1}= & \ j\phi^{2j}c[k(\mu+\sigma)+\mu]+j\phi^{2j+4}c\left[k(\mu+\sigma)+\sigma\right]
	+\phi^{2}a\left[j b(2k-2j+1)+2d\right]>0,\\
		A_{2}= & \ j ab \left[(k-j)(\phi^{4}+1)+1\right]  
				+2ad\phi^{2j+2} 
		 +j c (2k+1)\phi^{2j+2}(\mu+\sigma)>0,\\
		A_{3}= & -a\left[j\left(\phi^{2}-1\right)-k\phi^{2}+k+1\right]+\phi^{2j}\left[\phi^{2}(k(\mu+\sigma)+\sigma)-k(\mu+\sigma)-\mu\right]<0,
	\end{align*}

\noindent
where $A_{3}<0$ can be shown
by performing the variable change $\phi^{j}=x\in(0,\phi)$.
Hence, the zeros of the derivative are given by
 the zeros of $A_{1}-A_{2}$. First, notice that: 
	\[
	\lim_{\phi\rightarrow1}\frac{\partial\mathcal{S}^{{\rm FO}}}{\partial\phi}=j\left[\frac{j(\sigma-\mu)}{2k\sigma+\sigma}+\mu\left(\frac{1}{\sigma}+\frac{1}{\sigma-1}\right)\right]>0.
	\]
	This means, together with
	 $\lim_{\phi\rightarrow1}\mathcal{S}^{{\rm FO}}=0$ 
	 in \eqref{SFElim},
	 that 
	 \begin{align} &
	 \mathcal{S}^{{\rm FO}}<0 \quad \mbox{for~}\phi=1-\epsilon
	 \label{SFEnear1}\end{align}
	with $\epsilon>0$ arbitrarily
	small,   while
	$\mathcal{S}^{{\rm FO}}>0$ for $\phi=\epsilon$. 
	Thus, by the Intermediate Value Theorem, there exists at least
	one local sustain point
	$\phi_{j}^{c}\in(0,1)$ such that $\mathcal{S}^{{\rm FO}}(\phi_{j}^{c})=0$.
	\end{proof}
	
{\bf Uniqueness of local sustain point
(Proof of Lemma \ref{Sustainable_Full_Agglomeration_NEW}(ii)):}
	In the study of the uniqueness,
we employ the no-black-hole condition $\sigma-1>\mu$. 

Let us define $F(\phi)\equiv A_{1}-A_{2}$. Notice that:
	\begin{align}
		& F(0)=  -j(k-j+1)ab<0, \quad
			F(1)= 0, \quad F^{\prime}(0)=0, \quad F^{\prime}(1)=0, \nonumber \\
		& F^{\prime\prime}(0)= 
			 2a\left[j(-2j+2k+1)b+2d \right]>0,
			 \nonumber \\
		& F^{\prime\prime}(1)=  8j\left\{ j d 
		\sigma+\mu\left[j+\sigma(-j+4k+2)-2k-1\right]\right\} >0.
			\label{eq:criteria1}
	\end{align}
	Then, given the six results in (\ref{eq:criteria1}), the Intermediate
	Value Theorem establishes that $F(\phi)$ has at least one zero for
	$\phi\in(0,1)$ for any $k$ and any $j \ (\neq 0)$. 
	
	Next, we demonstrate the uniqueness of the critical point for
	$j\in\{1,\ldots,6\}$, by showing Lemma~\ref{Existence-Lemma-F}.
	This establishes that $\mathcal{S}^{{\rm FO}}$
	has at most one turning point (a minimum), which implies that the
	critical point is unique and corresponds to a local sustain point.
	
	\begin{lemma}\label{Existence-Lemma-F}
	$F(\phi)$ has exactly
	one zero for $\phi\in(0,1)$ for each $j\in\{1,\ldots,6\}$.
	\end{lemma}
	\begin{proof}
	The proof consists of using Descartes' rule of signs after reordering
	the monomials of $F(\phi)$ and determine the maximum number of positive
	real roots of $F(\phi)$. 
	For this purpose, we compute 
    $F_{j}(\phi)=F(\phi)|_{i=j}$ $(j=1,\ldots,6)$:
	\begin{align*}
		& F_{1}(\phi)= \left(\phi^{2}-1\right)^{2}\left\{ \phi^{2}c\left[k(\mu+\sigma)+\sigma\right]-kab\right\} , \\
		&F_{2}(\phi)=  2\left(\phi^{2}-1\right)^{2}\left\{ \phi^{4}c\left[k(\mu+\sigma)+\sigma\right]+\phi^{2}\mu a-(k-1)ab\right\} , \\
		&F_{3}(\phi)=  \left(\phi^{2}-1\right)^{2} 
		 \left\{ 3\phi^{6}c\left[k(\mu+\sigma)+\sigma\right]
		 +\phi^{4}a(3\mu+d) 
  +\phi^{2}a\langle 2\mu-b\rangle 
		-3(k-2)ab\right\} , \\
 &F_{4}(\phi)=  2\left(\phi^{2}-1\right)^{2}
	 \left\{ 2\phi^{8}c\left[k(\mu+\sigma)+\sigma\right]
		 +\phi^{6}a(2\mu+d) +2\phi^{4}\mu a 
		  +\phi^{2}a\langle \mu-b \rangle -2(k-3)ab\right\} , \\
 &	F_{5}(\phi)= \left(\phi^{2}-1\right)^{2}
		\left\{ 5\phi^{10}c\left[k(\mu+\sigma)+\sigma\right]\right.
		 +\phi^{8}a(5\mu+3d)
		 +\phi^{6}a(5\mu+d) \\
		&\phantom{F_{5}(\phi)= \left(\phi^{2}-1\right)^{2} ~~}
		 +\phi^{4}a \langle 4\mu-b \rangle 
		 \left. +\phi^{2}a \langle 2\mu-3b \rangle
		 -5(k-4)ab \right\},\\
& F_{6}(\phi)=  2\left(\phi^{2}-1\right)^{2}
  \left\{ 3\phi^{12}c \left[k(\mu+\sigma)+\sigma\right]\right.
		 +\phi^{10}a(3\mu+2d)
		 +\phi^{8} a (3\mu+d) +3\phi^{6}\mu a \\
		& \phantom{2\left(\phi^{2}-1\right)^{2}2
		\left(\phi^{2}-1\right)^{2}}
		\left. +\phi^{4}a \langle 2\mu-b \rangle 
		 +\phi^{2}a \langle \mu-2b \rangle 
		 -3(k-5)ab \right\} ,
	\end{align*}
	\noindent
where $\langle \cdot \rangle$ denotes a term that can take both signs
dependent on the values of the parameters $\mu$ and $b=\sigma-\mu-1$.
	We focus on the polynomial of $\phi^2$ in the curly bracket
	 $\{\cdot\}$
	 and, in turn, to note that the sign change of the coefficients
	 of the polynomial occurs only once for each 
$F_i$ $(i \in \{1,\ldots,6\})$.
  Note that, for $F_1$, 
	the coefficient for $\phi^2$ is positive (+) and the coefficient 
	for the constant term is negative ($-$); 
	accordingly, the sign changes once from positive to negative. 
For $F_3$,
	the signs for the coefficients are +, +, $\pm$, $-$; accordingly, 
	the sign change occurs once
	($\langle \cdot \rangle$ denotes a term 
	that becomes positive or negative depending on the
	values of the parameters therein).  
For $F_5$,
	although the two terms $\langle 4\mu-b \rangle$
	and $\langle 2\mu-3b \rangle$ can take both signs, 
     $\langle 2\mu-3b \rangle$ becomes negative first
    as $b$ increases	
	since $\langle 4\mu-b \rangle>\langle 2\mu-3b \rangle$.
	Hence, the sign change occurs once 
	for any values of $\mu$ and $b$.
	Other cases of $F_2$, $F_4$, and $F_6$ 
	can be treated similarly.
This concludes the proof.
	\end{proof}

{\bf Stability of full agglomeration
(Proof of Proposition \ref{Stability_Full_Agglomeration}):}
From $\mathcal{S}^{{\rm FO}}<0$ for $\phi=1-\epsilon$
in \eqref{SFEnear1}, with $\epsilon>0$ low enough,
we see that the full agglomeration is stable in the neighborhood of $\phi=1$
and becomes unstable at a local sustain point
 $\phi_i^{\rm s}$ for some place $i$.
 Such local sustain point serves as the sustain point
 and is defined as  
 $\phi^{\rm s}=\max_i\phi_i^{\rm s}$.
The full agglomeration is stable for $\phi\in (\phi^{\rm s},1)$.
If each local sustain point is unique,
the full agglomeration is unstable for $\phi\in (0,\phi^{\rm s})$.

\subsection{The PFSU model}\label{The PFSU model Appendix}

{\bf Derivation of $v_i$
(Proof of \eqref{eq:indirectutilitiesPFSUmodel0}
and \eqref{eq:indirectutilitiesPFSUmodel}):}
		For the PFSU model, we have $b_{i}=1+\lambda_{i}$, which is a term represented in the wage (\ref{eq:generalwageequation}).
	The wage in the central place $i=0$ is given by
	$w_{0}=\dfrac{\alpha}{\sigma}\left(2k+2\right)$.
	On the other hand, substituting $\xi = \alpha$ and $b_{0} = 2$ into  (\ref{eq:nominal wage final}) yields the wage in place $j \neq 0$:
	\[
	w_{\pm j}=\dfrac{\alpha}{\sigma}\left[ \phi^{j}(k+2)+\phi^{-j}(k-j)+\frac{\phi^{j}-\phi^{-j}}{\phi^{2}-1}\right]
	\quad	(j=1,\ldots,k).
	\]
	Substituting Eq. (\ref{AppDeltaFullAgg}) and the above wage equations into Eq. (\ref{IndiUtiPFSU}) yields the indirect utilities for the PFSU model 
	in places $i=0$
	 and $\pm j \ (1\le j \le k)$:
	\begin{align*}
		& v_{0}=\dfrac{\alpha}{\sigma}\left(2k+2\right)-\gamma\log 2 + \xi,
		\\
		& v_{\pm j}=\dfrac{\alpha}{\sigma}
			\left[ \phi^{j}(k+2)+\phi^{-j}(k-j)+\frac{\phi^{j}-\phi^{-j}}{\phi^{2}-1}\right]+\dfrac{\alpha j}{\sigma-1}\log\phi + \xi.
	\end{align*}


\setcounter{figure}{0}
\setcounter{table}{0}
\setcounter{equation}{0}
 \renewcommand{\thefigure}{C\arabic{figure}}
 \renewcommand{\thesection}{C}
 \renewcommand{\thetable}{C\arabic{table}}
\renewcommand{\theequation}{C\arabic{equation}}

\section{The MT model}\label{AnalysisMT}

The MT model \cite{Murata.Thisse.2005}
is presented to supplement the discussion of the main text.
This model employs a simplified yet reasonable specification to express
the interplay between commuting costs and inter-locational transport
costs.
We closely follow \citet{Takayama.2020}, who extended
the MT model to several places.

\subsection{Modeling}

The internal structure of each place is assumed to be one-dimensional
and featureless, except that there is a given CBD; the city expands
symmetrically around the origin. There are only skilled workers, 
who can choose their residential place $i\in N$ and location
$x$ in that place, where the CBD is located at $x=0$.
Land endowment equals unity everywhere in a place, and agents are assumed to consume one unit of land inelastically. 
The opportunity cost of land is normalized to zero in every place. Then, the city spreads
over the interval ${\mathcal{X}_i}\equiv[-\lambda_{i}/2,\lambda_{i}/2${]},
where $-\lambda_{i}/2$ and $\lambda_{i}/2$ denote the city boundaries.

Commuting costs take an iceberg form. Specifically, the effective
labor supply of a worker located at $x$ is given by $s(x)=1-4\theta|x|$,
for $x\in{\mathcal{X}_i}$,
where $\theta\in[0,1/2)$ is the commuting
rate that ensures that $s(x)\geq0$ for all
$x\in{\mathcal{X}_i}$ and
for all $i\in N$. Then, the total effective labor supply at the CBD
of place $i$ is given by:
\[
L_{i}=\int_{x_{i}\in\mathcal{X}_{i}}s(x)dx=\lambda_{i}(1-\theta\lambda_{i}).
\]

\noindent
Let $r_{i}(x)$ denote the land rent at $x$. A residential equilibrium implies that the wage net of commuting costs and land rent must be equal across locations:
$s(x)w_{i}-r_{i}(x)=s(\lambda_{i}/2)w_{i}$.
We thus have:
$r_{i}(x)=2\theta(\lambda_{i}-2|x|)w_{i}$,
which implies an aggregate land rent at place $i$ of:
\[
R_{i}=\int_{x_{i}\in\mathcal{X}_{i}}r_{i}(x)dx=\theta w_{i}\lambda_{i}^{2}.
\]

\noindent
Finally, since land is locally owned, the income of a worker in a place
$i$ at any location $x$ is given by
$y_{i}=(1-\theta\lambda_{i})w_{i}$.
The following upper-tier utility function gives preferences
$u_{i}^{\rm MT}=C_{i}$.
The remaining difference between the MT model and the footloose entrepreneur
models is that each firm now requires a fixed input of one worker
and a variable input of one worker per output 
of the produced manufactured goods.

Following \citet{Takayama.2020},
we reach the short-run equilibrium
whereby firms earn zero profits, which yields the following wage equation
for city $i$:
\[
L_{i} w_{i}=
L_{i}w_{i}^{1-\sigma}\sum_{j\in N}
\dfrac{\phi_{ij}L_{j}w_{j}}{\sum_{m\in N}
L_{m}w_{m}^{1-\sigma}\phi_{jm}}.
\]

\noindent
We normalize $\sum_{j\in N}\lambda_{j}w_{i}=w_{k}$.

In the state of full agglomeration to the central city $0$,
the indirect utility in a potential city $i$ is given by
$v_{i}=\zeta\Delta_{i}^{\frac{1}{\sigma-1}}y_{i}$,
where $y_{i}=(1-\theta \lambda_{i})w_{i}$, $\Delta_{i}
=\sum_{j\in N}L_{j}w_{j}^{1-\sigma}\phi_{ij}$
and $\zeta>0$ is a constant.
The indirect utilities
for $\bm{\lambda} = \bm{\lambda}^{\mathrm{FA}}_0$
are given by
\begin{align}
& v_{0} =\zeta(1-\theta)^{\frac{\sigma}{\sigma-1}},
\quad
v_{\pm j} =\zeta(1-\theta)^{\frac{1}{\sigma-1}}\phi^{\frac{j}{\sigma}\frac{2\sigma-1}{\sigma-1}}
\quad (1\le j \le k).
\label{eq:indirectutilitiesMTmodel}
\end{align}

\subsection{Theoretical details}\label{The MT model Appendix}

{\bf Derivation of $v_i$
(Proof of \eqref{eq:indirectutilitiesMTmodel}):}
In the MT model, we have the wage equation:
\[
w_{i}L_{i}=\sum_{p\in N}\dfrac{\phi_{ip}L_{i}w_{i}^{1-\sigma}}{\sum_{m\in N}L_{m}w_{m}^{1-\sigma}\phi_{pm}}L_{p}w_{p},
\]

\noindent
where $L_{i}=\lambda_{i}(1-\theta\lambda_{i})$.
Let us normalize $\sum_{j}\lambda_{j}w_{i}=w_{0}$,
where $w_{0}$ is the wage in the central city at full agglomeration.
The wage at a potential city at $i$
for any $L_{i}$ is given
by:
\begin{align*}
w_{i}^{\sigma}=\sum_{j\in N}
\dfrac{\phi_{ij}L_{j}w_{j}}
{\sum_{m \in N}L_{m}w_{m}^{1-\sigma}\phi_{jm}}
\iff
w_{i}=\left(\sum_{j \in N}
\dfrac{\phi_{ij}L_{j}w_{j}}
{\sum_{m \in N}L_{m}w_{m}^{1-\sigma}\phi_{jm}}
\right)^{1/\sigma}.
\end{align*}

\noindent
At the full agglomeration,
for which
the location $j\neq 0$ have zero population,
the wage becomes
$w_{j}=\phi^{\frac{j}{\sigma}}w_{0}$.
The indirect utility is given by
$v_{i}=\zeta\Delta_{i}^{\frac{1}{\sigma-1}}y_{i}$,
where $y_{i}=(1-\theta \lambda_{i})w_{i}$,
$\Delta_{i}=\sum_{p \in N}L_{p}w_{p}^{1-\sigma}\phi_{ip}$
and $\zeta>0$ is a constant.
Using $w_{j}=\phi^{\frac{j}{\sigma}}w_{0}$,
we obtain
\[
v_{0}=\zeta(1-\theta)^{\frac{\sigma}{\sigma-1}};
\quad
v_{\pm j}=\zeta(1-\theta)^{\frac{1}{\sigma-1}}\phi^{\frac{j}{\sigma}\frac{2\sigma-1}{\sigma-1}},
\quad 1\le j \le k.
\]

\noindent
{\bf Unique existence of sustain point:}
We have
\begin{equation}
\mathcal{S}^{\rm MT}\equiv v_{\pm j}-v_{0}=(1-\theta)^{\frac{1}{\sigma-1}}\phi^{\frac{j}{\sigma}\frac{2\sigma-1}{\sigma-1}}-(1-\theta)^{\frac{\sigma}{\sigma-1}}.
\label{eq:sustain pointMT}
\end{equation}
We have
$\mathcal{S}^{\rm MT}<0$ for $\phi=0$ and
$\mathcal{S}^{\rm MT}>0$ for $\phi=1$.
The solution to $\mathcal{S}^{\rm MT}=0$
with \eqref{eq:sustain pointMT} is
$\phi=\phi_{j}^{\rm s}\equiv(1-\theta)^{\frac{(\sigma-1)\sigma}{j(2\sigma-1)}}$.

The critical point is uniquely determined for each $j$ so that the
local sustain point is also uniquely determined. 
Further, we have
\[
\frac{\partial\phi_{j}^{\rm s}}{\partial j}=-\frac{(\sigma-1)\sigma\log(1-\theta)(1-\theta)^{-\frac{(\sigma-1)\sigma}{j-2j\sigma}}}{j^{2}(2\sigma-1)}>0,
\]

\noindent which means that $\phi_{j}^{\rm s}$ is strictly increasing in $j$. Since full agglomeration is stable below the sustain point, we have that
$\delta_{\rm s}=1$ and
$\phi=\phi^{s}
\equiv(1-\theta)^{\frac{(\sigma-1)\sigma}{2\sigma-1}}$.
The sustain point in terms of commuting costs
is given by
$\theta=\theta^{s}\equiv 1-\phi^{\frac{1}{\sigma}\frac{2\sigma-1}{\sigma-1}}$.

\end{document}